\documentclass[a4paper,USenglish]{lipics-v2019}
\usepackage{format}
\usepackage{pierre_macros}
\usepackage{macros}
\pdfoutput=1 

\newif\iflong
\longtrue 

\iflong
\else
\renewcommand{\refappendix}[1]{\cite{longversion}}
\fi

\iflong
\title{Leaderless State-Machine Replication: Specification, Properties, Limits\\ (Extended Version)}
\else
\title{Leaderless State-Machine Replication: Specification, Properties, Limits}
\fi
\titlerunning{Leaderless State-Machine Replication: Specification, Properties, Limits}

\author{Tuanir Fran\c{c}a Rezende}{Telecom SudParis, France }{tuanir.franca-rezende@telecom-sudparis.eu} {}{}
\author{Pierre Sutra}{Telecom SudParis, France }{pierre.sutra@telecom-sudparis.eu}{}{}
\authorrunning{T.\,F. Rezende, P. Sutra}
\Copyright{T.\,F. Rezende, P. Sutra}
\ccsdesc[100]{General and reference~Performance}
\ccsdesc[100]{Software and its engineering~Distributed systems organizing principles}
\ccsdesc[100]{Theory of computation~Distributed computing models}
\keywords{Fault Tolerance, State Machine Replication, Consensus}

\iflong
\hideLIPIcs
\else
\fi
\EventEditors{To be defined.}
\EventNoEds{2}
\EventLongTitle{submitted to DISC 2020}
\EventShortTitle{DISC 2020}
\EventAcronym{DISC}
\EventYear{2020}
\EventDate{October 13-15, 2020}
\EventLocation{Freiburg, Germany}
\EventLogo{}
\SeriesVolume{42}
\ArticleNo{23}

\begin{document}

\maketitle

\begin{abstract}
  Modern Internet services commonly replicate critical data across several geographical locations using state-machine replication (SMR).
  Due to their reliance on a leader replica, classical SMR protocols offer limited scalability and availability in this setting.
  To solve this problem, recent protocols follow instead a leaderless approach, in which each replica is able to make progress using a quorum of its peers.
  In this paper, we study this new emerging class of SMR protocols and states some of their limits.
  We first propose a framework that captures the essence of leaderless state-machine replication (Leaderless SMR).
  Then, we introduce a set of desirable properties for these protocols: (R)eliability, (O)ptimal (L)atency and (L)oad Balancing.
  We show that protocols matching all of the ROLL properties are subject to a trade-off between performance and reliability.
  We also establish a lower bound on the message delay to execute a command in protocols optimal for the ROLL properties.
  This lower bound explains the persistent chaining effect observed in experimental results.
\end{abstract}



\section{Introduction}
\labsection{introduction}

The standard way of implementing fault-tolerant distributed services is state-machine replication (SMR)~\cite{smr}.
In SMR, a service is defined by a deterministic state machine, and each process maintains its own local copy of the machine.
Classical SMR protocols such as Paxos \cite{paxos} and Raft \cite{raft} rely on a leader replica to order state-machine commands.
The leader orchestrates a growing sequence of agreements, or consensus, each defining the next command to apply on the state machine.
Such a scheme has however clear limitations, especially in a geo-distributed setting.
First, it increases latency for clients that are far away from the leader.
Second, as the leader becomes a bottleneck or its network gets slower, system performance decreases.
Last, this approach harms availability because when the leader fails the whole system cannot serve new requests until an election takes place.

To sidestep the above limitations, a new class of leaderless protocols has recently emerged \cite{mencius,epaxos,caesar,atlas,alvin,clockrsm}.
These protocols allow any replica to make progress as long as it is able to contact enough of its peers.
Mencius \cite{mencius} pioneered this idea by rotating the ownership of consensus instances.
Many other works have followed, and in particular the Egalitarian Paxos (EPaxos) protocol \cite{epaxos}.
As Generalized Paxos \cite{gpaxos}, EPaxos orders only non-commuting, aka. conflicting, state-machine commands.
To this end, the protocol maintains at each replica a directed graph that stores the execution constraints between commands.
Execution of a command proceeds by linearizing the graph of constraints.
In the common case, EPaxos executes a command after two message delays if the fast path was taken, that is, if the replicas spontaneously agree on the constraints, and four message delays otherwise.

\paragraph{Problem statement}
Unfortunately the latency of EPaxos may raise in practice well above four message delays.
To illustrate this point, we ran an experimental evaluation of EPaxos, Paxos and Mencius in Google Cloud Platform.
The results are reported in \reffigure{latency:cdf}, where we plot the cumulative distribution function (CDF) of the command latency for each protocol.
In this experiment, the system spans five geographical locations distributed around the globe, and each site hosts 128 clients that execute \emph{no-op} commands in closed-loop.
\reffigure{latency:distance} indicates the distance between any two sites.
The conflict rate among commands varies from 0\% to 30\%.%
\footnote{
  Each command has a key and any two commands conflict, that is they must be totally ordered by the protocol, when they have the same key.
  When a conflict rate $\rho$ is applied, each client picks key $42$ with probability $\rho$, and a unique key otherwise.
}
We measure the latency from the submission of a command to its execution (at steady state).

Two observations can be formulated at the light of the results in \reffigure{latency}.
First, the tail of the latency distribution in EPaxos is larger than for the two other protocols and it increases with the conflict rate.
Second, despite Mencius clearly offering a lower median latency, it does not exhibit such a problem.

\begin{figure}
  \begin{minipage}{.6\textwidth}
  \begin{subfigure}{\textwidth}
  \centering
  \fontsize{13}{15}\selectfont
  \scalebox{.51}{\input{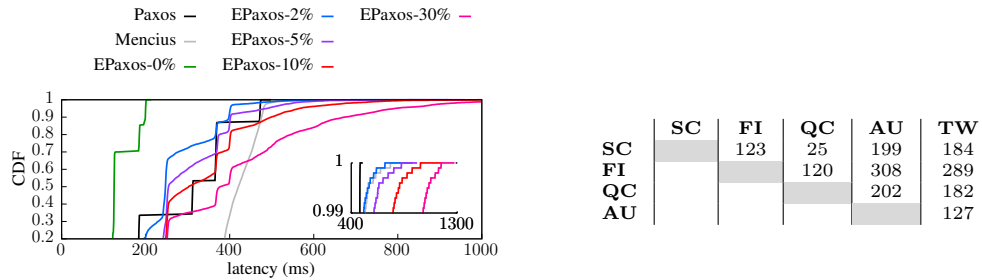}}
  \caption{
    \labfigure{latency:cdf}%
    Latency distribution when varying the conflict rate.
  }
  \end{subfigure}
  \end{minipage}
  \begin{minipage}{.4\textwidth}
    \vspace{5em}%
    \begin{subfigure}{\textwidth}
      \centering
      \scriptsize
      \begin{tabular}{l|c|c|c|c|c}
        & \textbf{SC} & \textbf{FI} & \textbf{QC} & \textbf{AU} & \textbf{TW} \\
        \textbf{SC} & \cellcolor{black!15} & 123         & 25                  & 199                        & 184      \\
        \textbf{FI} &             & \cellcolor{black!15} & 120                  & 308                         & 289      \\
        \textbf{QC} &             &                      & \cellcolor{black!15} & 202                        & 182      \\
        \textbf{AU} &             &                      &                      & \cellcolor{black!15}       & 127      \\
      \end{tabular}
      \vspace{2.5em}%
      \caption{
        \labfigure{latency:distance}
        Ping distance between sites (in ms).
      }
    \end{subfigure}
  \end{minipage}
  \caption{
    \labfigure{latency}
    Performance comparison of EPaxos, Paxos and Mencius
    -- \emph{5 sites: South Carolina (SC), Finland (FI), Canada (QC), Australia (AU), Taiwan (TW, leader); 128 clients per site; \emph{no-op} service}.
  }
  \vspace{-1em}
\end{figure}

\paragraph{Contributions}
In this paper, we provide a theoretical framework to understand and explain the above phenomena.
We study in-depth this new class of leaderless state-machine replication (Leaderless SMR) protocols and state some of their limits.

\paragraph{Paper Outline}
We recall the principles of state-machine replication (\refsection{smr}).
Then, we define Leaderless SMR and deconstruct it into basic building blocks (\refsection{leaderless}).
Further, we introduce a set of desirable properties for Leaderless SMR: (R)eliability, (O)ptimal (L)atency and (L)oad Balancing.
Protocols that match all of the ROLL properties are subject to a trade-off between performance and reliability.
More precisely, in a system of $n$ processes, the ROLL theorem (\refsection{roll}) states that Leaderless SMR protocols are subject to the inequality $2F+f-1 \leq n$, where $n-F$ is the size of the fast path quorum and $f$ is the maximal number of tolerated failures.
%
A protocol is ROLL-optimal when $F$ and $f$ cannot be improved according to this inequality.
We establish that ROLL-optimal protocols are subject to a chaining effect that affect their performance (\refsection{chaining}).
As \epaxos is ROLL-optimal and Mencius not, the chaining effect explains the performance results observed in \reffigure{latency}.
We discuss the implications of this result (\refsection{discussion}) then put our work in perspective (\refsection{relatedwork}) before closing (\refsection{conclusion}).

\section{State machine replication}
\labsection{smr}

State-machine replication (SMR) allows a set of distributed processes to construct a linearizable shared object.
The object is defined by a deterministic state machine together with a set of commands.
Each process maintains its own local replica of the machine.
An SMR protocol coordinates the execution of commands applied to the state machine, ensuring that the replicas stay in sync.
This section recalls the fundamentals of SMR, as well as its generalization that leverages the commutativity of state-machine commands.

\subsection{System model}
\labsection{smr:model}

We consider the standard model of wait-free computation in a distributed message-passing system where processes may fail-stop \cite{flp}.
In \cite{CT96}, the authors extend this framework to include failure detectors.
This paper follows such a model of distributed computation.
Further details appear in \refappendix{model}.

\subsection{Classic SMR}
\labsection{smr:classic}


State machine replication is defined over a set of $n \geq 2$ processes $\procSet$ using a set $\cmdSet$ of state-machine commands.
Each process $p$ holds a log, that is a totally ordered set of entries that we assume unbounded.
Initially, each entry in the log is empty (i.e., $log_p[i]=\bot$ for $i \in \naturalSet$), and over time it may include one state-machine command.
The operator ($\log_p \append \cmdc$) \emph{appends} command $\cmdc$ to the log, assigning it to the next free entry.

Commands are submitted by the processes that act as proxies on behalf of a set of remote clients (not modeled).
A process takes the step $\submit(\cmdc)$ to submit command $\cmdc$ for inclusion in the log.
Command $\cmdc$ is \emph{decided} once it enters the log at some position $i$.
It is executed against the state machine when all the commands at lower positions ($j < i$) are already executed.
When the command is executed, its response value is sent back to the client.
For simplicity, we shall consider that two processes may submit the same command.

When the properties below hold during every execution, the above construct ensures that the replicated state machine implements a linearizable shared object.
\begin{description}
\item[\textbf Validity:]
  A command is decided once and only if it was submitted before.
\item[\textbf Stability:]
  If $\log_p[i] = \cmdc$ holds at some point in time, it is also true at any later time.
\item[\textbf Consistency:]
  For any two processes $p$ and $q$, if $\log_p[i]$ and $\log_q[i]$ are both non-empty, then they are equal.
\end{description}

\subsection{Generic SMR}
\labsection{smr:generic}

In their seminal works, Pedone and Schiper \cite{gb} and concurrently Lamport \cite{gpaxos} introduce an alternative approach to Classic SMR.
They make the key observation that if commands submitted to the state machine commute, then there is no need to order them.
Leveraging this, they replace the totally-ordered log used in Classic SMR by a partially-ordered one.
We call this approach Generic SMR.

Two commands $\cmdc$ and $\cmdd$ do not commute when for some state $s$, applying $\cmdc \cmdd$ to $s$ differs from applying $\cmdd \cmdc$.
This means that either both sequences do not lead to the same state, or one of the two commands does not return the same response value in the two sequences.
Generic SMR relies on the notion of \emph{conflicts} which captures a safe over-approximation of the non-commutativity of two state-machine commands.
In what follows, conflicts are expressed as a binary, non-reflexive and symmetric relation $\conflict$ over $\cmdSet$.

In Generic SMR, each variable $\log_p$ is a partially ordered log, i.e., a directed acyclic graph \cite{gpaxos}.
In this graph, vertices are commands and any two conflicting commands have a directed edge between them.
We use $G.V$ and $G.E$ to denote respectively the vertices of some partially ordered log $G$ and its edges.
The append operator is defined as follows: $G \append \cmdc \equaldef (G.V \union \{\cmdc\}, G.E \union \{ (\cmdd,\cmdc): \cmdd \in G.V \land \cmdd \conflict \cmdc\}$.
A command is decided once it is in the partially ordered log.
As previously, it gets executed once all its predecessors are executed.

For correctness, Generic SMR defines a set of properties over partially ordered logs similar to Classic SMR.
Stability is expressed in close terms, using a prefix relation between the logs along time.
Consistency requires the existence of a common least upper bound over the partially ordered logs.

To state this precisely, consider two partially ordered logs $G$ and $H$.
$G$ is prefix of $H$, written $G \pref H$, when $G$ is a subgraph of $H$ and for every edge $(\cmda, \cmdb) \in H.E$, if $\cmdb \in G.V$ then $(\cmda,\cmdb) \in G.E$.
Given a set $\G$ of partially ordered logs, $H$ is an \emph{upper bound} of $\G$ iff $G \pref H$ for every $G$ in $\G$.
Two logs $G$ and $H$ are \emph{compatible} iff they have a common upper bound.%
\footnote{
  In \cite{gpaxos}, compatibility is defined in terms of least upper bound between two c-structs.
  For partially ordered logs, the definition provided here is equivalent.
}
By extension, a set $\G$ of partially ordered logs is compatible iff its elements are pairwise compatible.

Based on the above definitions, we may express Generic SMR using the set of properties below.
Validity is identical to Classic SMR and thus omitted.
\begin{description}
\item[\textbf Stability:]
  For any process $p$, at any given time $\log_p$ is prefix of itself at any later time.
\item[\textbf Consistency:]
  The set of all the partially ordered logs is always compatible.  
\end{description}

\section{Leaderless SMR}
\labsection{leaderless}


Some recent protocols \cite{mencius, epaxos} further push the idea of partially ordered log, as proposed in Generic SMR.
In a leaderless state-machine replication (Leaderless SMR) protocol, there is no primary process to arbitrate upon the ordering of commands.
Instead, any process may decide a command submitted to the replicated service.
A command is stable, and thus executable, once the transitive closure of its predecessors is known locally.
As this transitive closure can be cyclic, the log is replaced with a directed graph.

This section introduces a high-level framework to better understand Leaderless SMR.
In particular, we present the notion of dependency graph and explain how commands are decided.
With this framework, we then deconstruct several Leaderless SMR protocols into basic building blocks.
Further, three key properties are introduced: Reliability, Optimal Latency and Load Balancing.
These properties serve in the follow-up to establish lower bound complexity results for this class of protocols.

\subsection{Definition}
\labsection{leaderless:definition}

Leaderless SMR relies on the notion of dependency graph instead of partially ordered log as found in Generic SMR.
A \emph{dependency graph} is a directed graph that records the constraints defining how commands are executed.
For some command $\cmdc$, the incoming neighbors of $\cmdc$ in the dependency graph are its \emph{dependencies}.
As detailed shortly, the dependencies are executed either before or together with $\cmdc$.

In Leaderless SMR, a process holds two mapping: $\deps$ and $\phase$.
The mapping $\deps$ is a dependency graph storing a relation from $\cmdSet$ to $2^{\cmdSet} \union \{\bot,\top\}$.
For a command $\cmdc$, $\phase(\cmdc)$ can take five possible values: $\sPending$, $\sAbort$, $\sCommit$, $\sStable$ and $\sExecute$.
All the phases, except $\sExecute$, correspond to a predicate over $\deps$.

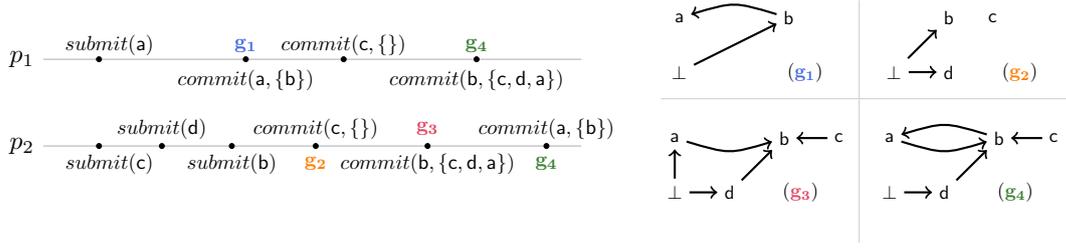
\begin{figure}[t]
  \centering
  \captionsetup{justification=centering}
  \begin{minipage}{.43\textwidth}
    \centering
    \tikzstyle{messageA} = [draw, -latex',Blue, shorten <=0.2em]
\tikzstyle{messageB} = [draw, -latex',Red, shorten <=0.2em]
\tikzstyle{lifeline} = [draw,Gray!60]
\begin{tikzpicture}[
    scale=0.92,
    every node/.style={transform shape},
    auto,
    shifttl/.style={shift={(-\shiftpoints,\shiftpoints)}},
    shifttr/.style={shift={(\shiftpoints,\shiftpoints)}},
    shiftbl/.style={shift={(-\shiftpoints,-\shiftpoints)}},
    shiftbr/.style={shift={(\shiftpoints,-\shiftpoints)}},
  ]
  \large

  \node [] (p1) {$p_1$};
  \node at (0, -1.25) (p2) {$p_2$};

  \path [lifeline] (p1) -- +(8,0);
  \path [lifeline] (p2) -- +(8,0);   


  \node[scale=0.8] at (1.25, 0.2) () {\color{Black}$submit(\cmda)$};
  \filldraw[black] (1.1, 0) circle (1pt) node[anchor=west] {};

  \node[scale=0.8] at (3.2, -0.3) () {\color{Black}$commit(\cmda,\{\cmdb\})$};
  \filldraw[black] (3.2, 0) circle (1pt) node[anchor=west] {};
  \node[scale=0.8] at (3.2, 0.2) () {$\grapha$};

  \node[scale=0.8] at (4.6, 0.2) () {\color{Black}$commit(\cmdc,\{\})$};
  \filldraw[black] (4.6, 0) circle (1pt) node[anchor=west] {};

  \node[scale=0.8] at (6.5, -0.3) () {\color{Black}$commit(\cmdb,\{\cmdc,\cmdd,\cmda\})$};
  \filldraw[black] (6.5, 0) circle (1pt) node[anchor=west] {};
  \node[scale=0.8] at (6.5, 0.2) () {$\graphd$};


  \node[scale=0.8] at (1.25,-1.5) () {\color{Black}$submit(\cmdc)$};
  \filldraw[black] (1.1,-1.25) circle (1pt) node[anchor=west] {};

  \node[scale=0.8] at (2,-1) () {\color{Black}$submit(\cmdd)$};
  \filldraw[black] (2,-1.25) circle (1pt) node[anchor=west] {};

  \node[scale=0.8] at (3,-1.5) () {\color{Black}$submit(\cmdb)$};
  \filldraw[black] (3,-1.25) circle (1pt) node[anchor=west] {};

  \node[scale=0.8] at (4.2,-1) () {\color{Black}$commit(\cmdc,\{\})$};
  \filldraw[black] (4.2,-1.25) circle (1pt) node[anchor=west] {};
  \node[scale=0.8] at (4.2, -1.5) () {$\graphb$};

  \node[scale=0.8] at (5.8,-1.5) () {\color{Black}$commit(\cmdb,\{\cmdc,\cmdd,\cmda\})$};
  \filldraw[black] (5.8,-1.25) circle (1pt) node[anchor=west] {};
  \node[scale=0.8] at (5.8, -1) () {$\graphc$};

  \node[scale=0.8] at (7.5, -1) () {\color{Black}$commit(\cmda,\{\cmdb\})$};
  \filldraw[black] (7.5,-1.25) circle (1pt) node[anchor=west] {};
  \node[scale=0.8] at (7.5, -1.5) () {$\graphd$};
\end{tikzpicture}
  \end{minipage}
  \hfill
  \begin{minipage}{0.38\textwidth}
    \centering
    \captionsetup{justification=centering}
    \begin{tabular}{@{}c|c}
      \noalign{\global\arrayrulewidth=0.1mm}
      \arrayrulecolor{gray!40}
      \begin{subfigure}[t]{0.45\textwidth}
        \centering
        \begin{tikzpicture}
 [
    scale=0.72,
    every node/.style={transform shape},
    auto,
    shifttl/.style={shift={(-\shiftpoints,\shiftpoints)}},
    shifttr/.style={shift={(\shiftpoints,\shiftpoints)}},
    shiftbl/.style={shift={(-\shiftpoints,-\shiftpoints)}},
    shiftbr/.style={shift={(\shiftpoints,-\shiftpoints)}},
  ]
  \begin{scope}[<-,thick]
    \node (a) at (0,0) {$\cmda$};
    \node (b) at (2,0) {$\cmdb$};
    \node (bot) at (0,-1) {$\bot$};
    \draw (b) -- (bot);
    \node (g1) at (2.3,-1) {$(\grapha)$};
    \draw (a) .. controls +(1,.3) .. (b);
  \end{scope}
\end{tikzpicture}
      \end{subfigure}
      &
      \begin{subfigure}[t]{0.45\textwidth}
        \centering
        \begin{tikzpicture}
   [
    scale=0.72,
    every node/.style={transform shape},
    auto,
    shifttl/.style={shift={(-\shiftpoints,\shiftpoints)}},
    shifttr/.style={shift={(\shiftpoints,\shiftpoints)}},
    shiftbl/.style={shift={(-\shiftpoints,-\shiftpoints)}},
    shiftbr/.style={shift={(\shiftpoints,-\shiftpoints)}},
  ]
  \begin{scope}[<-,thick]
    \node (b) at (1,0) {$\cmdb$};
    \node (c) at (1.8,0) {$\cmdc$};
    \node (d) at (1,-1) {$\cmdd$};
    \node (bot) at (0,-1) {$\bot$};
    \draw (b) -- (bot);
    \draw (d) -- (bot);
    \node (g2) at (2.3,-1) {$(\graphb)$};
  \end{scope}
\end{tikzpicture}
      \end{subfigure}\\
      \hline
      \begin{subfigure}[t]{0.45\textwidth}
        \centering
        \begin{tikzpicture}
   [
    scale=0.72,
    every node/.style={transform shape},
    auto,
    shifttl/.style={shift={(-\shiftpoints,\shiftpoints)}},
    shifttr/.style={shift={(\shiftpoints,\shiftpoints)}},
    shiftbl/.style={shift={(-\shiftpoints,-\shiftpoints)}},
    shiftbr/.style={shift={(\shiftpoints,-\shiftpoints)}},
  ]
  \begin{scope}[<-,thick]
    \node (a) at (0,0) {$\cmda$};
    \node (b) at (2,0) {$\cmdb$};
    \node (c) at (3,0) {$\cmdc$};
    \node (d) at (1,-1) {$\cmdd$};
    \node (bot) at (0,-1) {$\bot$};
    \draw (b) .. controls +(-1,-.3) .. (a);
    \draw (b) -- (c);
    \draw (b) -- (d);
    \draw (a) -- (bot);
    \draw (d) -- (bot);
    \node (g3) at (2.3,-1) {$(\graphc)$};
  \end{scope}
\end{tikzpicture}
      \end{subfigure}
      &
      \begin{subfigure}[t]{0.45\textwidth}
        \centering
        \begin{tikzpicture}
   [
    scale=0.72,
    every node/.style={transform shape},
    auto,
    shifttl/.style={shift={(-\shiftpoints,\shiftpoints)}},
    shifttr/.style={shift={(\shiftpoints,\shiftpoints)}},
    shiftbl/.style={shift={(-\shiftpoints,-\shiftpoints)}},
    shiftbr/.style={shift={(\shiftpoints,-\shiftpoints)}},
  ]
  \begin{scope}[<-,thick]
    \node (a) at (0,0) {$\cmda$};
    \node (b) at (2,0) {$\cmdb$};
    \node (c) at (3,0) {$\cmdc$};
    \node (d) at (1,-1) {$\cmdd$};
    \node (bot) at (0,-1) {$\bot$};
    \draw (a) .. controls +(1,.3) .. (b);
    \draw (b) .. controls +(-1,-.3) .. (a);
    \draw (b) -- (c);
    \draw (b) -- (d);
    \draw (d) -- (bot);
    \node (g3) at (2.3,-1) {$(\graphd)$};
    \node (blank) at (0,.6) {};
  \end{scope}
\end{tikzpicture}
      \end{subfigure}
    \end{tabular}
  \end{minipage}
  \caption{
    An example run of Leaderless SMR -- \emph{(left) processes $p_1$ and $p_2$ submit respectively the commands $\{\cmda\}$ and $\{\cmdb,\cmdc,\cmdd\}$; (right) the dependencies graphs formed at the two processes.}
    \labfigure{dep2}
  }
\end{figure}

Initially, for every command $\cmdc$, $\deps(\cmdc)$ is set to $\bot$.
This corresponds to the $\sPending$ phase.
When a process decides a command $\cmdc$, it changes the mapping $\deps(\cmdc)$ to a non-$\bot$ value.
Operation $\commit(\cmdc,D)$ assigns $D$ taken in $2^{\cmdSet}$ to $\deps(\cmdc)$.
Command $\cmdc$ gets aborted when $\deps(\cmdc)$ is set to $\top$.
In that case, the command is removed from any $\deps(\cmdd)$ and it will not appear later on.
Let $\deps^*(\cmdc)$ be the transitive closure of the $\deps$ relation starting from $\{\cmdc\}$.
Command $\cmdc$ is stable once it is committed and no command in $\deps^*(\cmdc)$ is pending.

\reffigure{dep2} depicts an example run of Leaderless SMR that illustrates the above definitions.
In this run, process $p_1$ submits command $\cmda$, while $p_2$ submits in order $\cmdc$, $\cmdd$ then $\cmdb$.
The timeline in \reffigure{dep2} indicates the timing of these submissions.
It also includes events during which process $p_1$ and $p_2$ commits commands.
For some of these events, we depict the state of the dependency graph at the process (on the right of \reffigure{dep2}).
As an example, the two processes obtain the graph $\graphd$ at the end of the run.
In this graph, $\cmda$, $\cmdb$ and $\cmdc$ are all committed, while $\cmdd$ is still pending.
We have $\deps(\cmda)=\{\cmdb\}$ and $\deps(\cmdb)=\{\cmda,\cmdd,\cmdc\}$, with both $\deps^*(\cmda)$ and $\deps^*(\cmdb)$ equal to $\{\cmda,\cmdb,\cmdc,\cmdd\}$.
Only command $\cmdc$ is stable in $\graphd$.

Similarly to Classic and Generic SMR, Leaderless SMR protocols requires that validity holds.
In addition, processes must agree on the value of $\deps$ for stable commands and conflicting commands must see each other.
More precisely,

\begin{description}
\item[Stability:]
  For each command $\cmdc$, there exists $D$ such that if $\cmdc$ is stable then $\deps(\cmdc)=D$.
\item[Consistency:]
  If $\cmda$ and $\cmdb$ are both committed and conflicting, then $\cmda \in \deps(\cmdb)$ or $\cmdb \in \deps(\cmda)$.
\end{description}

A command $c$ gets executed once it is stable.
\refalg{execute} describes how this happens in Leaderless SMR.
To execute command $\cmdc$, a process first creates a set of commands, or \emph{batch}, $\beta$ that execute together with $\cmdc$.
This grouping of commands serves to maintain the following invariant:

\begin{invariant}
  \labinv{del}
  Consider two conflicting commands $\cmdc$ and $\cmdd$.
  If $p$ executes a batch of commands containing $\cmdc$ before executing $\cmdd$, then $d \notin \deps^*(c)$.
\end{invariant}

Satisfying \refinv{del} implies that if some command $\cmdd$ is in batch $\beta$, then $\beta$ also contains its transitive dependencies (\refline{execute:4} in \refalg{execute}).
Inside a batch, commands are ordered according to the partial order $\eorder$ (\refline{execute:5}).
Let $<$ be a canonical total order over $\cmdSet$.
Then, $\cmdc \eorder \cmdd$ holds iff
\begin{inparaenum}
\item $\cmdc \in \deps^*(\cmdd)$ and $\cmdd \notin \deps^*(\cmdc)$; or
\item $\cmdc \in \deps^*(\cmdd)$, $\cmdd \in \deps^*(\cmdc)$ and $\cmdc < \cmdd$.
\end{inparaenum}
Relation $\eorder$ defines the \emph{execution order} at a process.
If there is a one-way dependency between two commands, Leaderless SMR plays them in the order of their transitive dependencies;
otherwise the algorithm breaks the tie using the arbitrary order $<$.
This guarantees the following invariant.

\begin{invariant}
  \labinv{del2}  
  Consider two conflicting commands $\cmdc$ and $\cmdd$.
  If $p$ executes $\cmdc$ before $\cmdd$ in the same batch, then $\cmdc \in \deps^*(\cmdd)$.
\end{invariant}

\begin{algorithm}[t]
 
  \caption{Executing command $\cmdc$ -- code at process $p$}
  \labalg{execute}

  \begin{algorithmic}[1]

    \begin{action}{$\execute(\cmdc)$} \labline{execute:1}
      \precondition $\phase(\cmdc) = \sStable$ \labline{execute:2}
      \effect \textbf{let} $\beta$ be the largest subset of $\deps^*(\cmdc)$ satisfying $\forall \cmdd \in \beta \ldotp \phase(\cmdd) \equals \stable$ \labline{execute:4}
      \effect \textbf{forall} {$\cmdd \in \beta$ {\bf ordered by $\eorder$}} \labline{execute:5}
      \effect \hspace{1em} $\phase(\cmdd) \assign \sExecute$ \labline{execute:6}
    \end{action}
    
  \end{algorithmic}

\end{algorithm}


Generic and Leaderless SMR are strongly similar.
In fact, one may show that Generic SMR reduces to Leaderless SMR without requiring any message exchange.
This result is stated in \reftheo{reduction} below, and a proof appears in \refappendix{proofs}.
Let us observe that such a reduction does not hold between Classic and Generic SMR.
Indeed, computing a total order on commuting commands would require processes to communicate.

\begin{restatable}[]{thm}{reduction}
  \labtheo{reduction}
  Generic SMR reduces to Leaderless SMR.
\end{restatable}

However, \reftheo{reduction} offers an incomplete picture of how the two abstractions compare in practice.
Indeed, because the dependency graph might be cyclic, Leaderless SMR does not compute an ordering over conflicting commands.
Instead, such commands must simply observe one another (Consistency property).
This fundamental difference explains the absence of a leader in this class of SMR protocols, a feature that we capture in the next section.

\subsection{Deciding commands}
\label{leaderless:deciding}

In Leaderless SMR, processes have to agree on the dependencies of stable commands.
Thus, a subsequent refinement leads to consider a family of consensus objects $(\cons_{\cmdc})_{\cmdc \in \cmdSet}$ for that purpose.
For some command $\cmdc$, processes use $\cons_{\cmdc}$ to decide either the dependencies of $\cmdc$, or the special value ($\top$) signaling that the command is aborted.
This agreement is driven by the command \emph{coordinator} ($\coordc$), a process initially in charge of submitting the command to the replicated state machine.
In a run during which there is no failure and the failure detector behaves perfectly, that is a \emph{nice run}, only $\coordc$ calls $\cons_{\cmdc}$.

To create a valid proposal for $\cons_{\cmdc}$, $\coordc$ relies on the dependency discovery service ($\dds$).
This shared object offers a single operation $\announce(\cmdc)$ that returns a pair $(D,\flag)$, where $D \in 2^{\cmdSet} \union \{\top\}$ and $\flag \in \{0,1\}$ is a flag.
When the return value is in $2^{\cmdSet}$, the service suggests to commit the command.
Otherwise, the command should be aborted.
When the flag is set, the service indicates that a spontaneous agreement occurs.
In such a case, the coordinator can directly commit $\cmdc$ with the return value of the $\dds$ service and bypass $\cons_c$;
this is called \emph{a fast path}.
A \emph{recovery} occurs when command $c$ is annonunced at a process which is not $\coordc$.

The $\dds$ service ensures two safety properties.
First, if two conflicting commands are announced, they do not miss each other.
Second, when a command takes the fast path, processes agree on its committed dependencies.

More formally, assume that $\announce_p(\cmdc)$ and $\announce_q(\cmdcp)$ return respectively $(D,\flag)$ and $(D',\flag')$ with $D \in 2^{\cmdSet}$.
Then, the properties of the $\dds$ service are as follows.

\begin{description}
\item[Visibility:]
  If $\cmdc \conflict \cmdcp$ and $D' \in 2^{\cmdSet}$, then $\cmdc \in D'$ or $\cmdcp \in D$.
\item[Weak Agreement:]
  If $\cmdc = \cmdcp$ and $\flag = \true$, then $D' \in 2^{\cmdSet}$ and for every $\cmdd \in D \xor D'$, every invocation to $\announce_r(\cmdd)$ returns $ (\top,\any)$.
\end{description}

To illustrate these properties, consider that no command was announced so far.
In that case $(\emptySet,\true)$ is a valid response to $\announce(\cmdc)$.
If $\coordc$ is slow, then a subsequent invocation of $\announce(\cmdc)$ may either return $\emptySet$, or a non-empty set of dependencies $D$.
However in that case, because the fast path was taken by the coordinator, all the commands in $D$ must eventually abort.

\begin{algorithm}[t]
 
  \caption{Deciding a command $\cmdc$ -- code at process $p$}
  \labalg{deciding}

  \begin{algorithmic}[1]

    \begin{action}{$\submit(\cmdc)$} \labline{commit:1}
      \precondition $p = \coord(\cmdc) \lor \coord(\cmdc) \in \fd$ \labline{commit:2}
      \effect $(D,\flag) \assign \dds.\announce(\cmdc)$ \labline{commit:3}
      \effect \textbf{if} $\flag = \false$ \textbf{then} $D \assign \cons_{\cmdc}.propose(D)$ \labline{commit:4}
      \effect $\deps(\cmdc) \assign D$ \labline{commit:5}
      \effect $\sendTo{\cmdc,\deps(\cmdc)}{\procSet \setminus \{p\}}$ \labline{commit:6}
    \end{action}

    \\
    
    \begin{when}{$\recv(\cmdc,D)$} \labline{commit:7}
      \effect $\deps(\cmdc) \assign D$ \labline{commit:8}
    \end{when}
    
  \end{algorithmic}

\end{algorithm}

Based on the above decomposition of Leaderless SMR, \refalg{deciding} depicts an abstract protocol to decide a command.
This algorithm uses a family of consensus objects ($(\cons_{\cmdc})_{\cmdc \in \cmdSet}$), a dependency discovery service ($\dds$) and a failure detector ($\fd$) that returns a set of suspected processes.
To submit a command $\cmdc$, a process announces it then retrieves a set of dependencies.
This set is proposed to $\cons_{\cmdc}$ if the fast path was not taken (\refline{commit:4}).
The result of the slow or the fast path determines the value of the local mapping $\deps(\cmdc)$ to commit or abort command $\cmdc$.
Notice that such a step may also be taken when a process receives a message from one of its peers (\refline{commit:7}).

During a nice run, the system is failure-free and the failure detector service behaves perfectly.
As a consequence, only $\coordc$ may propose a value to $\cons_c$ and this value gets committed.
In our view, this feature is the \emph{key characteristic} of Leaderless SMR.

Below, we establish the correctness of \refalg{deciding}.
A proof appears in \refappendix{proofs}.

\begin{restatable}[]{thm}{deciding}
  \labtheo{deciding}
  \refalg{deciding} implements Leaderless SMR.
\end{restatable}

\subsection{Examples}
\labsection{leaderless:examples}

To illustrate the framework introduced in the previous sections, we now instantiate well-known Leaderless SMR protocols using it.

\paragraph{Rotating coordinator}
For starters, let us consider a rotating coordinator algorithm (e.g., \cite{spin}).
In this class of protocols, commands are ordered \emph{a priori} by some relation $\ll$.
Such an ordering is usually defined by timestamping commands at each coordinator and breaking ties with the process identities.
When $\coordc$ calls $\dds.\announce(c)$, the service returns a pair $(D,\false)$, where $D$ are all the commands prior to $\cmdc$ according to $\ll$.
Upon recovering a command, the $\dds$ service simply suggests to abort it.

\paragraph{Clock-RSM}
This protocol \cite{clockrsm} improves on the above schema by introducing a fast path.
It also uses physical clocks to speed-up the stabilization of committed commands.
Once a command is associated to a timestamp, its coordinator broadcasts this information to the other processes in the system.
When it receives such a message, a process waits until its local clock passes the command's timestamp to reply.
Once a majority of processes have replied, the \dds service informs the coordinator that the fast path was taken.

\paragraph{Mencius}
The above two protocols require a committed command to wait all its predecessors according to $\ll$.
Clock-RSM propagates in the background the physical clock of each process.
A command gets stable once the clocks of all the processes is higher than its timestamp.
Differently, Mencius \cite{mencius} aborts prior pending commands at the time the command is submitted.
In detail, $\announce(c)$ first approximates $D$ as all the commands prior to $\cmdc$ according to $\ll$.
Then, command $\cmdc$ is broadcast to all the processes in the system.
Upon receiving such a message, a process $q$ computes all the commands $\cmdd$ smaller than $\cmdc$ it is coordinating.
If $\cmdd$ is not already announced, $q$ stores that $\cmdd$ will be aborted.
Then, $q$ sends $\cmdd$ back to $\coordc$ that removes it from $D$.
The $\dds$ service returns $(D,f)$ with $f$ set to $\true$ if $\coordc$ received a message from everybody.
Upon recovering $\cmdc$, if the command was received the over-approximation based on $\ll$ is returned together with the flag $\false$.
In case $\cmdc$ is unknown, the $\dds$ service suggests to abort it.

\paragraph{EPaxos}
In \cite{epaxos}, the authors present Egalitarian Paxos (EPaxos), a family of efficient Leaderless SMR protocols.
For simplicity, we next consider the variation which does not involve sequence numbers.
To announce a command $\cmdc$, the coordinator broadcasts it to a quorum of processes.
Each process $p$ computes (and records) the set of commands $D_p$ conflicting with $\cmdc$ it has seen so far.
A call to $\announce(\cmdc)$ returns $(\union_p D_p, \flag)$, with $\flag$ set to $\true$ iff processes spontaneously agree on dependencies (i.e., for any $p,q$, $D_p=D_q$).
When a process in the initial quorum is slow or a recovery occurs, $\cmdc$ is broadcast to everybody.
The caller then awaits for a majority quorum to answer and returns $(D,\false)$ such that if at least $\frac{f+1}{2}$ processes answer the same set of conflicts for $\cmdc$, then $D$ is set to this value (with $n=2f+1$).
Alternatively, if at least one process knows $\cmdc$, the union of the response values is taken.
Otherwise, the $\dds$ service suggests to abort $c$.

\paragraph{Caesar}
To avoid cycles in the dependency graph, Caesar \cite{caesar} orders commands using logical timestamps.
Upon submitting a command $\cmdc$, the coordinator timestamps it with its logical clock then it executes a broadcast.
As with EPaxos, when it receives $\cmdc$ a process $p$ computes the conflicting commands $D_p$ received so far.
Then, it awaits until there is no conflicting command $\cmdd$ with a higher timestamp than $\cmdc$ such that $\cmdc \notin \deps(\cmdd)$.
If such a command exists, $p$ replies to the coordinator that the fast path cannot be taken.
The $\dds$ service returns $(\union_p D_p, \flag)$, where $\flag=\true$ iff no process disables the fast path.

The above examples show that multiple implementations are possible for Leaderless SMR.
In the next section, we introduce several properties of interest to characterize them.

\subsection{Core properties}
\labsection{leaderless:properties}

State machine replication helps to mask failures and asynchrony in a distributed system.
As a consequence, a first property of interest is the largest number of failures (parameter $f$) tolerated by a protocol.
After $f$ failures, the protocol may not guarantee any progress.%
\footnote{
  When $f$ failures occur, the system configuration must change to tolerate subsequent ones.
  If data is persisted (as in Paxos \cite{paxos}), the protocol simply stops when more than $f$ failures occurs and awaits that faulty processes are back online.
}

\begin{description}
\item[(\emph{Reliability})]
  In every run, if there are at most $f$ failures, every submitted command gets eventually decided at every correct process.
\end{description}

Leaderless SMR protocols exploit the absence of contention on the replicated service to boost performance.
In particular, some protocols are able to execute a command after a single round-trip, which is clearly optimal \cite{consensusBounds}.
To ensure this property, the fast path is taken when there is no concurrent conflicting command.
Moreover, the command stabilizes right away, requiring that the $\dds$ service returns only submitted commands.

\begin{description}
\item[(\emph{Optimal Latency})]
  During a nice run, every call to $\announce(c)$ returns a tuple $(D,\flag)$ after two message delays such that
  \begin{inparaenum}
  \item if there is no concurrent conflicting command to $c$, then $\flag$ is set to $\true$,
  \item $D \in 2^{\cmdSet}$, and
  \item for every $d \in D$, $d$ was announced before.
  \end{inparaenum}
\end{description}

The replicas that participate to the fast path vary from one protocol to another.
Mencius use all the processes.
On the contrary, EPaxos solely contact $\floor{\frac{3n}{4}}$ of them (or equivalently, $f+\frac{f+1}{2}$ when $n=2f+1$).
For some command $\cmdc$, a \emph{fast path quorum} for $\cmdc$ is any set of $n-F$ replicas that includes the coordinator of $\cmdc$. 
Such a set is denoted $\fquorums(c)$ and formally defined as $\{ Q \mid Q \subseteq \Pi \land coord(c) \in Q \: \land |Q| \geq n - F \}$.
A protocol has the \emph{Load Balancing} property when it may freely choose fast path quorum to make progress.

\begin{description}
\item[(\emph{Load Balancing})]
  During a nice run, any fast path quorum in $\fquorums(c)$ can be used to announce a command $\cmdc$.
\end{description}

The previous properties are fomally defined in \refappendix{roll:prop}.
\reftab{properties} indicates how they are implemented by well-known leaderless protocols.
The columns 'Reliability' and 'Load Balancing' detail respectively the maximum number of failures tolerated by the protocol and the size of the fast path quorum.
Notice that by CAP \cite{cap}, we have $F,f \leq \floor{\frac{n-1}{2}}$ when the protocol matches all of the properties.
%
\reftab{properties} also mentions the optimality of each protocol with respect to the ROLL theorem.
This theorem is stated in the next section and establishes a trade-off between fault-tolerance and performance in Leaderless SMR.

\begin{table}[t]
  \centering
  \begin{tabular}{c|c c c|c} 
    & \multicolumn{4}{c}{\textbf{\emph{Properties}}} \\ [0.5ex]
    \textbf{\emph{Protocols}} & \makecell{Load Balancing\\ ($n-F$)} & \makecell{Reliability\\ ($f$)} & \makecell{Optimal\\ Latency} & ROLL-optimal \\ 
    \hline
    \hline
    Rotating coord.& 0 & \MIN & \NO & \NO \\
    Clock-RSM \cite{clockrsm} & $n$ & \MIN & \NO & \NO \\\
    Mencius \cite{mencius} & $n$ & \MIN & \YES & \NO \\
    Caesar \cite{caesar} & $\ceil{\frac{3n}{4}}$ & \MIN & \YES & \NO \\
    EPaxos \cite{epaxos} & \LMAJ & \MIN & \YES & if $n=2f+1$ \\
    Alvin \cite{alvin} & \LMAJ & \MIN & \YES & if $n=2f+1$ \\
    Atlas \cite{atlas} & $\floor{\frac{n}{2}}+f$ & any & \YES & \makecell{if $n \in 2\naturalSet \union \{3\} \land f=1$}
  \end{tabular}
  \vspace{1em}
  \caption{
    \labtab{properties}%
    The properties of several leaderless SMR protocols
    --
    \MIN stands for a minority of replicas ($\floor{\frac{n-1}{2}}$),
    \MAJ a majority ($\ceil{\frac{n+1}{2}}$), and
    \LMAJ a large majority ($\floor{\frac{3n}{4}}$).
  }
  \vspace{-1.5em}
\end{table}

\section{The ROLL theorem}
\labsection{roll}

Reliability, Optimal Latency and Load Balancing are called collectively the ROLL properties.
These properties introduce the parameters $f$ and $F$ as key characteristics of a Leaderless SMR protocol.
Parameter $f$ translates the reliability of the protocol, stating that progress is guaranteed only if less than $f$ processes crash.
Parameter $F$ captures its scalability since, any quorum of $n - F$ processes may be used to order a command.
An ideal protocol should strive to minimize $n-F$ while maximizing $f$.

Unfortunately, we show that there is no free-lunch and that an optimization choice must be made.
The ROLL theorem below establishes that $2F+f-1 \leq n$ must hold.
This inequality captures that every protocol must trade scalability for fault-tolerance.
EPaxos \cite{epaxos} and Atlas \cite{atlas} illustrate the two ends of the spectrum of solutions (see \reftab{properties}).
EPaxos supports that any minority of processes may fail, but requires large quorums.
Atlas typically uses small fast path quorums ($\floor{\frac{n}{2}}+f$), but exactly handles at most $f$ failures.

Below, we state the ROLL theorem and provide a sketch of proof illustrated in \reffigure{roll}.
A formal treatment appears in \refappendix{roll}.


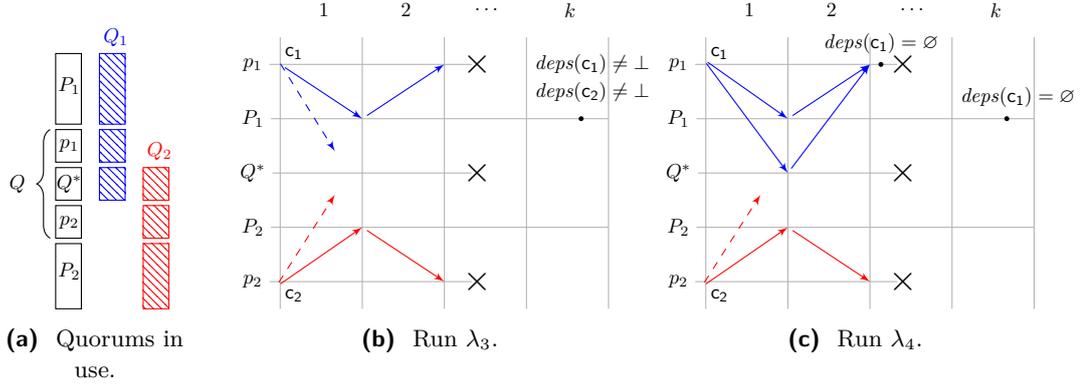
\begin{figure}[tbp]
  \centering
  \captionsetup{justification=centering} 
  \begin{subfigure}[t]{0.2\textwidth}
    \centering
    \begin{tikzpicture}[
    scale=0.72,
    every node/.style={transform shape},
    auto,
    shifttl/.style={shift={(-\shiftpoints,\shiftpoints)}},
    shifttr/.style={shift={(\shiftpoints,\shiftpoints)}},
    shiftbl/.style={shift={(-\shiftpoints,-\shiftpoints)}},
    shiftbr/.style={shift={(\shiftpoints,-\shiftpoints)}},
  ]
  \large

  \draw[Black] (1,-1.7) rectangle (1.47,-3);

  \draw[Black] (1,-3.1) rectangle (1.47,-3.7);
  \draw[Black] (1,-3.8) rectangle (1.47,-4.4);
  \draw[Black] (1,-4.5) rectangle (1.47,-5.1);

  \node at (1.25,-2.3) () {\color{Black}$P_1$};
  \node at (1.25,-3.4) () {\color{Black}$p_1$};
  \node at (1.25,-4.1) () {\color{Black}$Q^*$};
  \node at (1.25,-4.8) () {\color{Black}$p_2$};

  \draw[Black] (1,-5.2) rectangle (1.47,-6.4);
  \node at (1.25,-5.7) () {\color{Black}$P_2$};
  \node at (2.1,-1.4) () {\color{Blue}$Q_1$};
  \draw[Blue, pattern=north west lines, pattern color = Blue] (1.8,-1.7) rectangle (2.27,-3);
  \draw[Blue, pattern=north west lines, pattern color = Blue] (1.8,-3.1) rectangle (2.27,-3.7);
  \draw[Blue, pattern=north west lines, pattern color = Blue] (1.8,-3.8) rectangle (2.27,-4.4);

  \node at (2.9,-3.5) () {\color{Red}$Q_2$};
  \draw[Red, pattern=north west lines, pattern color = Red ] (2.6,-3.8) rectangle (3.07,-4.4);
  \draw[Red, pattern=north west lines, pattern color = Red ] (2.6,-4.5) rectangle (3.07,-5.1);
  \draw[Red, pattern=north west lines, pattern color = Red ] (2.6,-5.2) rectangle (3.07,-6.4);

  \draw [decorate,decoration={brace,amplitude=5pt,mirror}] (0.9,-3.1) -- (0.9,-5.1);
  \node at (0.3,-4.1) () {\color{Black}$Q$};

\end{tikzpicture}
    \caption{
      \labfigure{roll:a} 
      Quorums in use.
    }
    \label{fig:roll1}
  \end{subfigure}
  \hfill%
  \begin{subfigure}[t]{0.37\textwidth}
    \centering
    \tikzstyle{messageA} = [draw, -latex',Blue, shorten <=0.2em]
\tikzstyle{messageB} = [draw, -latex',Red, shorten <=0.2em]
\tikzstyle{lifeline} = [draw,Gray!60]

\begin{tikzpicture}[
    scale=0.72,
    every node/.style={transform shape},
    auto,
    shifttl/.style={shift={(-\shiftpoints,\shiftpoints)}},
    shifttr/.style={shift={(\shiftpoints,\shiftpoints)}},
    shiftbl/.style={shift={(-\shiftpoints,-\shiftpoints)}},
    shiftbr/.style={shift={(\shiftpoints,-\shiftpoints)}},
  ]
  \large

  \node [] (p1) {$p_1$};
  \node [below of = p1] (p2) {$P_1$};
  \node [below of = p2] (p3) {$Q^*$};
  \node [below of = p3] (p4) {$P_2$};
  \node [below of = p4] (p5) {$p_2$};

  \path [lifeline] (p1) -- +(4,0);
  \path [lifeline] (p2) -- +(6.5,0);
  \path [lifeline] (p3) -- +(4,0);
  \path [lifeline] (p4) -- +(6.5,0);
  \path [lifeline] (p5) -- +(4,0);   

  \node at (4.1, 0) [scale=1.8] () {\color{Black}$\times$};
  \node at (4.1,-2) [scale=1.8] () {\color{Black}$\times$};
  \node at (4.1,-4) [scale=1.8] () {\color{Black}$\times$};
  
  \path [lifeline] (0.5,0.5) -- (0.5,-4.5);   
  \path [lifeline] (2,0.5) -- (2,-4.5);   
  \path [lifeline] (3.5,0.5) -- (3.5,-4.5);   
  \path [lifeline] (5,0.5) -- (5,-4.5);   
  \path [lifeline] (6.5,0.5) -- (6.5,-4.5);   

  \node at (0.75,0.2) () {\color{Black}$\cmd{1}$};
  \path[messageA] (0.5,0) -- (2,-1);
  \path[messageA] (2,-1) -- (3.5,0);
  \path[messageA,dashed] (0.45,0.1) -- (1.5,-1.6);

  \node at (0.75,-4.25) () {\color{Black}$\cmd{2}$};
  \path[messageB] (0.42,-4.1) -- (2,-3);
  \path[messageB] (2,-3) -- (3.5,-4);
  \path[messageB,dashed] (0.42,-4.1) -- (1.5,-2.4);

  \node at (1.3,1) () {\color{Black}$1$};
  \node at (2.8,1) () {\color{Black}$2$};
  \node at (4.3,1) () {\color{Black}$\cdots$};
  \node at (5.8,1) () {\color{Black}$k$};

  \filldraw[black] (6,-1) circle (1pt) node[anchor=west] {};
  \node at (6.2,0) {\color{Black}{$\deps(\cmd{1}) \neq \bot$}};
  \node at (6.2,-0.5) {\color{Black}{$\deps(\cmd{2}) \neq \bot$}};

\end{tikzpicture}
    \caption{
      \labfigure{roll:b} 
      Run $\run_3$.
    }
  \end{subfigure}
  \hfill%
  \begin{subfigure}[t]{0.37\textwidth}
    \centering
    \tikzstyle{messageA} = [draw, -latex',Blue, shorten <=0.2em]
\tikzstyle{messageB} = [draw, -latex',Red, shorten <=0.2em]
\tikzstyle{lifeline} = [draw,Gray!60]

\begin{tikzpicture}[
    scale=0.72,
    every node/.style={transform shape},
    auto,
    shifttl/.style={shift={(-\shiftpoints,\shiftpoints)}},
    shifttr/.style={shift={(\shiftpoints,\shiftpoints)}},
    shiftbl/.style={shift={(-\shiftpoints,-\shiftpoints)}},
    shiftbr/.style={shift={(\shiftpoints,-\shiftpoints)}},
  ]
  \large

  \node [] (p1) {$p_1$};
  \node [below of = p1] (p2) {$P_1$};
  \node [below of = p2] (p3) {$Q^*$};
  \node [below of = p3] (p4) {$P_2$};
  \node [below of = p4] (p5) {$p_2$};

  \path [lifeline] (p1) -- +(4,0);
  \path [lifeline] (p2) -- +(6.5,0);
  \path [lifeline] (p3) -- +(4,0);
  \path [lifeline] (p4) -- +(6.5,0);
  \path [lifeline] (p5) -- +(4,0);   

  \node at (4.1, 0) [scale=1.8] () {\color{Black}$\times$};
  \node at (4.1,-2) [scale=1.8] () {\color{Black}$\times$};
  \node at (4.1,-4) [scale=1.8] () {\color{Black}$\times$};
  
  \path [lifeline] (0.5,0.5) -- (0.5,-4.5);   
  \path [lifeline] (2,0.5) -- (2,-4.5);   
  \path [lifeline] (3.5,0.5) -- (3.5,-4.5);   
  \path [lifeline] (5,0.5) -- (5,-4.5);   
  \path [lifeline] (6.5,0.5) -- (6.5,-4.5);   

  \node at (0.75,0.2) () {\color{Black}$\cmd{1}$};
  \path[messageA] (0.45,0.1) -- (2,-1);
  \path[messageA] (2,-1) -- (3.5,0);
  \path[messageA] (0.45,0.1) -- (2,-2);
  \path[messageA] (2,-2) -- (3.5,0);

  \node at (0.75,-4.25) () {\color{Black}$\cmd{2}$};
  \path[messageB] (0.42,-4.1) -- (2,-3);
  \path[messageB] (2,-3) -- (3.5,-4);
  \path[messageB,dashed] (0.42,-4.1) -- (1.5,-2.4);

  \node at (1.3,1) () {\color{Black}$1$};
  \node at (2.8,1) () {\color{Black}$2$};
  \node at (4.3,1) () {\color{Black}$\cdots$};
  \node at (5.8,1) () {\color{Black}$k$};

  \node at (3.7,0.4) {\color{Black}{$\deps(\cmd{1}) = \emptySet$}};
  \filldraw[black] (3.7,0) circle (1pt) node[anchor=west] {};

  \filldraw[black] (6,-1) circle (1pt) node[anchor=west] {};
  \node at (6.2,-0.6) {\color{Black}{$\deps(\cmd{1}) = \emptySet$}};

\end{tikzpicture}
    \vspace{-1.3em} 
    \caption{
      \labfigure{roll:c}
      Run $\run_4$.
    }
  \end{subfigure}
  \caption{
    \labfigure{roll}
    Illustration of \reftheo{roll} -- \emph{slow messages are omitted}.
  }
  \vspace{-1em}
\end{figure}

\begin{restatable}[ROLL]{thm}{roll}
  \labtheo{roll}
  Consider an SMR protocol that satisfies the ROLL properties.
  Then, it is true that $2F + f - 1 \leq n$.
\end{restatable}

\begin{sproof}
  Our proof goes by contradiction, using a round-based reasoning.
  Let us assume a protocol $\mathcal{P}$ that satisfies all the ROLL properties with $2F + f - 1 > n$.
  Then, choose two non-commuting commands $\cmd{1}$ and $\cmd{2}$ in $\cmdSet$.

  As depicted in \reffigure{roll:a}, the distributed system is partitioned into three sets:
  $P_1$ and $P_2$ are two disjoints sets of $F-1$ processes, and the remaining $n - 2(F-1)$ processes form $Q$.
  The CAP impossibility result \cite{cap} tells us that $2F < n$.
  As a consequence, there exist at least two distinct processes $p_1$ and $p_2$ in $Q$.
  We define $Q_1$ and $Q_2$ as respectively $P_1 \union Q \setminus \{p_2\}$ and $P_2 \union Q \setminus \{p_1\}$.
  The set $Q^{*}$ equals $Q \setminus \{p_1,p_2\}$.

  Let $\run_1$ be a nice run that starts from the submission of $\cmd{1}$ by process $p_1$ during which only $Q_1$ take steps.
  Since $Q_1$ contains $n-F$ processes such a run exists by the Load Balancing property of $\mathcal{P}$.
  By Optimal Latency, this run lasts two rounds and $\deps(\cmd{1})$ is set to $\emptySet$ at process $p_1$.
  Similarly, we may define $\run_2$ a run in which $p_2$ announces command $\cmd{2}$ and in which only the processes in $Q_2$ participate.

  Then, consider a run $\run_3$ in which $p_1$ and $p_2$ submit concurrently commands $\cmd{1}$ and $\cmd{2}$.
  This run is illustrated in \reffigure{roll:b}.
  At the end of the first round, the processes in $P_1$ (respectively, $P_2$) receive the same messages as in $\run_1$ (resp., $\run_2$).
  At the start of the second round, they reply to respectively $p_1$ and $p_2$ as in $\run_1$ and $\run_2$.
  All the other messages sent in the first two rounds are arbitrarily slow.
  The processes in $Q$ crash at the end of the second round.
  By Reliability and as $f \geq \cardinalOf{Q}$, the commands $\cmd{1}$ and $\cmd{2}$ are stable in $\run_3$.  
  Let $k$ be the first round at which the two commands are stable at some process $p \in P_1 \union P_2$.

  We now build an admissible run $\run_4$ of $\mathcal{P}$ as follows.
  The failure pattern and failure detector history are the same as in $\run_3$.
  Commands $\cmd{1}$ and $\cmd{2}$ are submitted concurrently at the start of $\run_4$, as in $\run_3$.
  In the first two rounds, $P_1$ receives the same messages as in $\run_1$ while $P_2$ receives the same messages as in $\run_3$.
  The other messages exchanged during the first two rounds are arbitrarily slow.
  \reffigure{roll:c} depicts run $\run_4$.
  
  Observe that the following claims about $\run_4$ are true.
  First, (C1) for $p_1$, $\run_4$ is indistinguishable to $\run_1$ up to round 2.
  Moreover, (C2) for the processes in $(P_1 \union P_2)$, $\run_4$ is indistinguishable to $\run_3$ up to round $k$.
  From (C1), $\cmd{1}$ is stable at $p_1$ with $\deps(\cmd{1}) = \emptySet$.
  Claim (C2) implies that both $\cmd{1}$ and $\cmd{2}$ are stable at $p$ when round $k$ is reached.
  By the stability property of Leaderless SMR, process $p$ and $p_1$ decide the same dependencies for $\cmd{1}$, i.e., $deps(\cmd{1}) = \emptySet$.

  A symmetric argument can be made using run $\run_2$ and a run $\run_5$, showing that $p$ decides $deps(\cmd{2}) = \emptySet$ in $\run_3$.
  It follows that in $\run_3$, an empty set of dependencies is decided for both commands at process $p$;
  a contradiction to the Consistency property.
\end{sproof}

\reftheo{roll} captures an inherent trade-off between performance and reliability for ROLL protocols.
For instance, tolerating a minority of crashes, requires accessing at least $\floor{\frac{3n}{4}}$ processes.
This is the setting under which EPaxos operates.
On the other hand, if the protocol uses a plain majority quorum in the fast path, it tolerates at most one failure.

\subsection{Optimality}
\labsection{roll:optimality}

A protocol is \emph{ROLL-optimal} when the parameters $F$ and $f$ cannot be improved according to \reftheo{roll}.
In other words, they belong to the skyline of solutions \cite{skyline}.
As an example, when the system consists of 5 processes, there is a single such tuple $(F,f)=(2,2)$.
With $n=7$, there are two tuples in the skyline, $(2,3)$ and $(3,2)$.
The first one is attained by EPaxos, while Atlas offers the almost optimal solution $(3,1)$ (see \reftab{properties}).

For each protocol, \reftab{properties} lists the conditions under which ROLL-optimality is attained.
EPaxos and Alvin are both optimal under the assumption that $n=2f+1$.
Atlas adjusts the fast path quorums to the value of $f$, requiring $\floor{\frac{n}{2}}+f$ processes to participate.
This is optimal when $f=1$ and either $n$ is even or equals to $3$.
In the general case, the protocol is within $O(f)$ of the optimal value.
As it uses classical Fast Paxos quorums, Caesar is not ROLL-optimal.
This is also the case of protocols that contact all of the replicas to make progress, such as Mencius and Clock-RSM.
To the best of our knowledge, no protocol is optimal in the general case.


In the next section, we show that ROLL-optimality has a price.
More precisely, we establish that by being optimal, a protocol may create an arbitrarily long chain of commands, even during a nice run.
This chaining effect may affect adversely the performance of the protocol.
We discuss measures of mitigation in \refsection{discussion}.

\section{Chaining effect}
\labsection{chaining}


This section shows that a chaining effect may affect ROLL-optimal protocols.
It occurs when the chain of transitive dependencies of a command keeps growing after it gets committed.
This implies that the committed command takes time to stabilize, thus delaying its execution and increasing the protocol latency.

At first glance, one could think that this situation arises from the asynchrony of the distributed system.
As illustrated in \reffigure{latency}, this is not the case.
We establish that such an effect may occur during ``almost'' synchronous runs.

The remaining of this section is split as follows.
First, we define the notion of chain, that is a dependency-related set of commands.
A chain is live when its last command is not stable.
To measure how asynchronous a nice run is, we then introduce the principle of $k$-asynchrony.
A run is $\async{k}$ when some message is concurrent to both the first and last message of a sequence of $k$ causally-related messages.

At core, our result shows how to inductively add a new link to a live chain during an appropriate $\async{2}$ run of a ROLL-optimal protocol.

\subsection{Notion of chain}
\labsection{chaining:chain}

A chain is a sequence of commands $\cmd{i}\ldots\cmd{n}$ such that for any two consecutive commands $(\cmdi,\cmdin)$ in the chain, $\cmdin \in \deps(\cmdi)$ at some process.
Two consecutive commands $(\cmdc, \cmdd)$ in a chain form a \emph{link}.
For instance, in the dependency graph $\graphd$ (see \reffigure{dep2}), $\cmdc \cmdb \cmda$ is a chain.

We shall say that a chain is \emph{live} when its last command is not stable yet (at any of the processes).
In $\graphd$, this is the case of the chain $\cmdd \cmdb \cmda$, since command $\cmdd$ is still pending ($\deps(\cmdd)=\bot$).
When a chain is live, the last command in the chain has to wait to ensure a sound execution order across processes.
This increases the protocol latency.

\subsection{A measure of asynchrony}
\labsection{chaining:asynchrony}

In a synchronous system \cite{lynch}, processes executes rounds in lock-step.
During a round, the messages sent at the beginning are received at the end (provided there is no failure).
On the other hand, a partially synchronous system may delay messages for an arbitrary amount of time.
In this model, we propose to measure asynchrony by looking at the overlaps between the exchanges of messages.
The larger the overlap is, the more asynchronous is the run.

To illustrate this idea, consider the run depicted in \reffigure{chaining}.
During this run, a {\color{Crimson}red} message is sent from $p_5$ to $p_4$ (bottom left corner of the figure).
In the same amount of time $p_1$ sends a {\color{Blue}blue} message to $p_2$ which is followed by a {\color{OliveGreen}green} message to $p_4$.
To characterize such an asynchrony, we shall say that the run is $\async{2}$.
This notion is precisely defined below.

\begin{definition}[Path]
  A sequence of event $\rho = \send_p(m_1) \recv_q(m_1) \send_q(m_2) \ldots \recv_t(m_{k \geq 1})$ in a run is called a \emph{path}.
  We note $\rho[i]$ the $i$-th message in the path.
  The number of messages in the path, or its \emph{size}, is denoted $\cardinalOf{\rho}$.
\end{definition}

\begin{definition}[Overlapping]
  Two messages $m$ and $m'$ are overlapping when their respective events are concurrent.%
  \footnote{
    That is, neither $\recv(m)$ precedes $\send(m')$, nor $\recv(m')$ precedes $\send(m)$ in real-time.
  }
  By extension, a message $m$ overlaps with a path $\rho$ when it overlaps with both $\rho[1]$ and $\rho[\cardinalOf{\rho}]$.
\end{definition}

\begin{definition}[$k$-asynchrony]
  A run $\lambda$ is \async{k} when for every message $m$, if $m$ overlaps with a path $\rho$ then $\cardinalOf{\rho} \leq k$. 
\end{definition}


\subsection{Result statement}
\labsection{chaining:result}

The theorem below establishes that a ROLL-optimal protocol may create a live chain of arbitrary size during a \async{2} nice run.
The full proof appears in \refappendix{chaining}.

\tikzstyle{message} = [draw, -latex',shorten <=0.2em]
\tikzstyle{messageA} = [draw, -latex',Blue, shorten <=0.2em]
\tikzstyle{messageB} = [draw, -latex',Crimson, shorten <=0.2em]
\tikzstyle{messageC} = [draw, -latex',OliveGreen, shorten <=0.2em]
\tikzstyle{messageD} = [draw, -latex',Purple, shorten <=0.2em]
\tikzstyle{messageE} = [draw, -latex',Orange, shorten <=0.2em]
\tikzstyle{messageF} = [draw, -latex',Brown, shorten <=0.2em]
\tikzstyle{messageG} = [draw, -latex',Red, shorten <=0.2em]
\tikzstyle{lifeline} = [draw,Gray!60]
\tikzstyle{call} = [thick]
\newcommand{\sstate}[2]{\ensuremath{{\color{#1}{#2}}}}
\newcommand{\sstateOne}{\ensuremath{\sstate{Blue}{\bigstar}}}
\newcommand{\sstateTwo}{\ensuremath{\sstate{Crimson}{\blacklozenge}}}
\newcommand{\sstateThree}{\ensuremath{\sstate{Purple}{\blacktriangle}}}
\newcommand{\shiftpoints}{2pt}

\begin{figure}[t]
  \centering
  \begin{tikzpicture}[
      scale=0.75,
      every node/.style={transform shape},
      auto,
      shifttl/.style={shift={(-\shiftpoints,\shiftpoints)}},
      shifttr/.style={shift={(\shiftpoints,\shiftpoints)}},
      shiftbl/.style={shift={(-\shiftpoints,-\shiftpoints)}},
      shiftbr/.style={shift={(\shiftpoints,-\shiftpoints)}},
    ]
    \large

    \node [] (p1) {$p_1$};
    \node [below of = p1] (p2) {$p_2$};
    \node [below of = p2] (p3) {$p_3$};
    \node [below of = p3] (p4) {$p_4$};
    \node [below of = p4] (p5) {$p_5$};

    \path [lifeline] (p1) -- +(16,0);
    \path [lifeline] (p2) -- +(16,0);
    \path [lifeline] (p3) -- +(16,0);
    \path [lifeline] (p4) -- +(16,0);
    \path [lifeline] (p5) -- +(16,0);   

    \path[messageA] (1.2,0) -- (2.2,-1);
    \path[messageA] (1.2,0) -- (2.2,-2);
    \path[messageA] (2.2,-1) -- (3.2,0);
    \path[messageA] (2.2,-2) -- (3.2,0);

    \draw[Black,thick,dotted] (1,.5) -- (1.7,.5) -- (1.7,-.5) -- (1,-.5) -- (1,.5);
    \draw[Black,thick,dotted] (1.85,.5) -- (2.55,.5) -- (2.55,-2.2) -- (1.85,-2.2) -- (1.85,.5);
    \draw[Black,thick,dotted] (2.7,.5) -- (3.4,.5) -- (3.4,-.5) -- (2.7,-.5) -- (2.7,.5);

    \node at (1.35,.8) () {\color{Black}$S_1$};
    \node at (2.2,.8) () {\color{Black}$M_1$};
    \node at (3.05,.8) () {\color{Black}$R_1$};
    
    \path[messageB] (.5,-4) -- (1.5,-2);
    \path[messageB] (.5,-4) -- (4.5,-3);
    \path[messageB] (1.5,-2) -- (2.5,-4);
    \path[messageB] (4.5,-3) -- (5.5,-4);

    \path[messageC] (3,-1) -- (4,-3);
    \path[messageC] (4,-3) -- (5,-1);
    \path[messageC] (3,-1) -- (7,0);
    \path[messageC] (7,0) -- (8,-1);

    \path[messageD] (5,-2) -- (6,0);
    \path[messageD] (6,0) -- (7,-2);
    \path[messageD] (5,-2) -- (9,-4);
    \path[messageD] (9,-4) -- (10,-2);

    \path[messageE] (6,-3) -- (7,-4);
    \path[messageE] (7,-4) -- (8,-3);
    \path[messageE] (6,-3) -- (10,-1);
    \path[messageE] (10,-1) -- (11,-3);

    \path[messageF] (8.5,0) -- (9.5,-1);
    \path[messageF] (9.5,-1) -- (10.5,0);
    \path[messageF] (8.5,0) -- (12.5,-2);
    \path[messageF] (12.5,-2) -- (13.5,0);

    \draw (10.5,.5) .. controls (11.7,-1) and (11.7,-3) .. (10,-4.5);
    \node at (10.5,.8) () {\color{Black}$\Gamma_7$};
    
    \draw[Black,thick,dotted] (12.2,.5) -- (13.7,.5) -- (13.7,-2.2) -- (12.2,-2.2) -- (12.2,.5);
    \node at (13.05,.8) () {\color{Black}$\Gamma_7'$};

    \path[messageG] (10.8,-4) -- (11.8,-2);
    \path[messageG] (10.8,-4) -- (14.8,-3);
    \path[messageG,dashed] (11.8,-2) -- (12.8,-4);
    \path[messageG,dashed] (14.8,-3) -- (15.8,-4);    
  \end{tikzpicture}
  \caption{
    \labfigure{chaining}
    \reftheo{chaining} for $n=5$ and $k=7$.
    The chain
    $
    {\color{Red}\cmd{7}}
    {\color{Brown}\cmd{6}}
    {\color{Orange}\cmd{5}}
    {\color{Purple}\cmd{4}}
    {\color{OliveGreen}\cmd{3}}
    {\color{Crimson}\cmd{2}}
    {\color{Blue}\cmd{1}}
    $
    is formed in $\sigma_{7}$.
    Illustrating the steps {\color{Black}$S_1$}, {\color{Black}$M_1$} and {\color{Black}$R_1$} for command ${\color{Blue}\cmd{1}}$,
    the prefix $\Gamma_7$ of $\sigma_7$, and the steps $\Gamma_7'$. 
  }
  \vspace{-1.5em}
\end{figure}
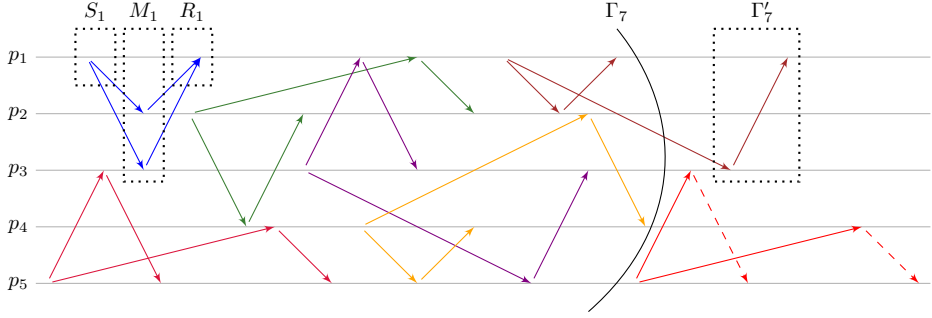

\begin{restatable}[Chaining Effect]{thm}{chaining}
  \labtheo{chaining}
  Assume a ROLL-optimal protocol $\mathcal{P}$.
  For any $k>0$, there exists a $\async{2}$ nice run of $\mathcal{P}$ containing a live chain of size $k$.
\end{restatable}

\begin{sproof}
  The theorem is proved by adding inductively a new link to a live chain of commands created during a nice run.
  It is illustrated in \reffigure{chaining} for a system of five processes when $k=7$.
  
  The proof is based on the following two key observations about ROLL-optimal protocols.
  First, during a nice run, the coordinator of a command never rotates.
  As a consequence, the return value of the \dds service at the coordinator is always the stable value of $\deps(c)$.
  Second, as the protocol satisfies the ROLL properties, a call to $\announce(c)$ consists of sending a set of requests to the fast path quorum and receiving a set of replies.
  As a consequence, its execution can be split into the steps $S_c M_c R_c$, where
  \begin{inparaenumorig}
  \item[($S_c$)] are the steps taken from announcing $c$ to the sending of the last request at the coordinator;
  \item[($R_c$)] are the steps taken by $\coord(c)$ after receiving the first replies until the announcement returns; and
  \item[($M_c$)] are the steps taken during the announcement of $c$ which are neither in $S_c$, nor in $R_c$.
  \end{inparaenumorig}%
  By Optimal Latency, this sequence of steps do not create pending messages. 
  As an illustration, the steps $S_1$, $M_1$ and $R_1$ taken to announce command ${\color{Blue}\cmd{1}}$ are depicted in \reffigure{chaining}.
  
  Leveraging the above two observations, the result is built inductively using a family of $k$ distinct commands $(\cmdi)_{i \in [1,k]}$.
  Each command is associated with a nice run ($\sigma_i$), a fast path quorum ($Q_i$), a subset of $f-1$ processes ($P_i$), and a process ($q_i$).
  
  Given a sequence of steps $\lambda$ and a set of processes $Q$, let us note $\lambda|Q$ the sub-sequence of steps by $Q$ in $\lambda$.
  We establish that at rank $i>0$ the following property $\mathfrak{P}(i)$ holds:
  There exists a $\async{2}$ run $\sigma_{i}$ of the form $\Gamma_i S_i (M_i|P_i) \Gamma'_i (M_i| Q_i \setminus P_i) R_i$ such that
  \begin{inparaenumorig}[\em(1)]
  \item $\proc(\Gamma_i') \inter Q_i = P_i$;
  \item every path in $\Gamma_i'$ is as most of size one;
  \item no message is pending in $\sigma_i$; and
  \item $\sigma_{i}$ contains a chain $\cmdi \cmdip \cdots \cmd{1}$.
  \end{inparaenumorig}%
  \reffigure{chaining} depicts the run $\sigma_{7}$, its prefix $\Gamma_7$ and the steps $\Gamma_7'$.

  Starting from $\mathfrak{P}(i)$, we establish $\mathfrak{P}(i+1)$ as follows.
  First we show that $\sigma_{i+1}$ as $\Gamma_{i+1} S_{i+1} (M_{i+1} | P_{i+1}) \Gamma'_{i+1} (M_{i+1} | Q_{i+1} \setminus P_{i+1}) R_{i+1}$, where 
  \begin{inparaenumorig}[]
  \item $\Gamma_{i+1} = \Gamma_i S_i (M_i | P_i) \Gamma'_i$, and
  \item $\Gamma'_{i+1} =  (M_i | Q_i \setminus P_i) R_i$
  \end{inparaenumorig}%
  is a nice run.

  At rank $i+1$, item (1) is proved with appropriate definitions of the quorums ($Q_i$ and $Q_{i+1}$), and the sub-quorum ($P_{i}$).
  For instance, in \reffigure{chaining}, the command $c_1$ and $c_2$ have respectively $\{p_1,p_2,p_3\}$ and $\{p_3,p_4,p_5\}$ for fast path quorums.
  The sub-quorum $P_2$ is set to the intersection of $Q_1$ and $Q_2$, that is $\{p_3\}$.
  Item (2) follows from the definition of $\Gamma_{i+1}'$.
  The Load-Balancing property implies that (3) holds.
  A case analysis can then show that $\sigma_{i+1}$ is \async{2}.
  It relies on the fact that the $(SMR)_{i+1}$ steps create no pending message and the induction property $\mathfrak{P}(i)$. 
  
  To prove that a new link was added, we show that $\sigma_{i+1}$ is indistinguishable to $\coordi$ to a run in which $\cmdin$ gets committed while missing $\cmdi$.
  Going back to \reffigure{chaining}, observe that the coordinator of ${\color{Brown}\cmd{6}}$ does not know that the replies of $p_2$ for command ${\color{Orange}\cmd{5}}$ causally precedes the replies of $p_4$ to ${\color{Red}\cmd{7}}$.
  As a consequence, it must add $\cmd{7}$ to the return value of $DDS.\announce({\color{Brown}\cmd{6}})$.

  Finally, to obtain a live chain of size $k$, it suffices to consider the prefix of $\sigma_{i+1}$ which does not contain the replies of the fast path quorum.
  In \reffigure{chaining}, this corresponds to omitting the dashed messages that contain the reply to the announcement of ${\color{Red}\cmd{7}}$
\end{sproof}

\section{Discussion}
\labsection{discussion}

Leaderless SMR offers appealing properties with respect to leader-driven approaches.
Protocols are faster in the best case, suffer from no downtime when the leader fails, and distribute the load among participants.
For instance, \reffigure{latency} shows that EPaxos is strictly better than Paxos when there is no conflict.
However, the latency of a command is strongly related to its dependencies in this family of SMR protocols.
Going back to \reffigure{latency}, the bivariate correlation between the latency of a command and the size of the batch with which it executes is greater than 0.7. 

Several approaches are possible to mitigate the chaining effect established in \reftheo{chaining}.
Moraru et al. \cite{epaxos} propose that pure writes (i.e., commands having no response value) return once they are committed.%
\footnote{
  In fact, it is possible to return even earlier, at the durable signal, that is once $f+1$ processes have received the command.
  To ensure linearizability, a later read must however wait for all the prior (conflicting or not) preceding writes.
}
In \cite{atlas}, the authors observe that as each read executes at a single process, they can be excluded from the computation of dependencies.
A third possibility is to wait until prior commands return before announcing a new one.
However, in this case, it is possible to extend \reftheo{chaining} by rotating the command coordinators to establish that a chain of size $n$ can form. 

In ROLL, the Load-Balancing and Optimal-Latency properties constrain the form of the \dds service.
More precisely, in a contention-free case, executing the service must consist in a back-and-forth between the command coordinator and the fast path quorum.
A weaker definition would allow some messages to be pending when $\announce$ returns. 
In this case, it is possible to sidestep the ROLL theorem provided that the system is synchronous:%
\footnote{
  Here, we consider that the synchronous system ensures that if $p$ sends $m$ then $m'$ to respectively $p$ and $q$ and $p$ is correct, $p$ receives $m$.
}
When replying to an announcement a process first sends its reply to the other fast path quorum nodes.
The fast path is taken by merging all of the replies.
Since the system is synchronous, a process recovering a command will retrieve all the replies at any node in the fast path quorum.
Note that under this weaker definition, the ROLL theorem (\reftheo{roll}) still applies in a partially synchronous model.
Moreover, a chaining effect (\reftheo{chaining}) is also possible, but it requires more asynchrony during a nice run. 

\section{Related work}
\labsection{relatedwork}


\paragraph{Protocols}
Early leaderless solutions to SMR include rotating coordinators and deterministic merge, aka. collision-fast, protocols.
We cover the first class of protocols in \refsection{leaderless:examples}.
In a collision-fast protocol \cite{deterministicmerge, collisionfast}, processes replicate an infinite array of vector consensus instances.
Each vector consensus corresponds to a round.
During a round, each process proposes a command (or a batch) to its consensus instance in the vector.
If the process is in late, its peers may take over the instance and propose an empty batch of commands.
Commands are executed according to their round numbers, applying an arbitrary ordering per round.
The size of the vector can change dynamically, adapting to network conditions and/or the application workload.
This technique is also used in Paxos Commit \cite{paxoscommit}.

When the ordering is fixed beforehand, processes must advance at the same pace.
To fix this issue, Mencius \cite{mencius} includes a piggy-back mechanism that allows a process to bail out its instances (i.e., proposing implicitly an empty batch).
Clock-RSM \cite{clockrsm} follows a similar schema, using physical clocks to bypass explicit synchronization in the good cases.

With the above protocols, commands still get delayed by slow processes.
Avoiding this so-called delayed commit problem \cite{mencius} requires to dynamically discover dependencies at runtime.
This is the approach introduced in Zieli{\'{n}}ski's optimistic generic broadcast \cite{optimisticgb} and EPaxos \cite{epaxos}.
Here, as well as in \cite{atlas}, replicas agree on a fully-fledged dependency graph.
Caesar \cite{caesar} uses timestamps to avoid cycles in the graph.
However, even in contention-free cases, committing a command can take two round-trips.
In our classification (see \reftab{properties}), this protocol does not have Optimal Latency.

\paragraph{Deconstruction}
In \cite{familyPODC16}, the authors introduce the dependency-set and map-agreement algorithms.
The two services allow respectively to gather dependencies and agree upon them.
A similar decomposition is proposed in \cite{bipaxos}.
Compared to these prior works, our framework includes the notion of fast path and distinguishes committed and stable commands.
An agreement between the processes is necessary only eventually and on the stable part of the dependency graph.
This difference allows to capture a wider spectrum of protocols.
Our dependency discovery service (\dds) is reminiscent of an adopt-commit object \cite{RbR} that allows processes to reach a weak agreement.
In our case, when the fast path flag is set, processes may disagree on at most the aborted dependencies of a command.

\paragraph{Complexity}
Multiple works study the complexity of consensus, the key underlying building block of SMR.
Lamport \cite{consensusBounds} proves several lower bounds on the time and space complexity of this abstraction.
The Hyperfast Learning theorem establishes that consensus requires one round-trip in the general case.
This explains why we call optimal protocols that return after two message delays.
The Fast Learning theorem requires that $n > 2F + f$.
This result explains the trade-off between fault-tolerance and performance in Fast Paxos \cite{fastpaxos}.
However, it does not readily apply to Leaderless SMR because only coordinator-centric quorums are fast in that case.
For instance, EPaxos is able to run with $F=1$ and $f=1$ in a 3-process system.
The ROLL theorem (\refsection{roll}) accurately captures this difference.

Traditional complexity measures for SMR and consensus (e.g., the latency degree \cite{schiper}) consider contention-free and/or perfectly synchronous scenarios.
In \cite{rachidSMR}, the authors study the complexity of SMR over long runs.
The paper shows that completing an SMR command can be more expensive than solving a consensus instance.
Their complexity measure is different from ours and given in terms of synchronous rounds.
In \refsection{chaining}, we show that in an almost synchronous scenario, contention may create arbitrarily long chains in Leaderless SMR.
We discuss mitigation measures in \refsection{discussion}.

\section{Conclusion}
\labsection{conclusion}

This paper introduces a framework to decompose leaderless state-machine replication (Leaderless SMR) protocols.
The framework allows to break down representative protocols into two simple building blocks: a dependency discovery service and a consensus service.
We then define a set of desirable properties for Leaderless SMR: (R)eliability, (O)ptimal (L)atency and (L)oad balancing.
Protocols matching all of these properties satisfy the inequality $2F + f - 1 \leq n$, where $n$ is number of processes, $f$ the maximum number of failures tolerated by the protocol, and $n-F$ the size of the fast path quorum.
Further, we establish that protocols that optimally solve this inequality suffer from a chaining effect.
This effect explains the tail latency of some Leaderless SMR protocols in real-world deployments.

\vspace{2em}
\paragraph{Acknowledgments}
The authors thank Vitor Enes and Alexey Gotsman for fruitful discussions on Leaderless SMR.
This research is partly funded by the ANR RainbowFS project and the H2020 CloudButton project.



\newpage
\bibliography{biblio}

\appendix

\iflong
\newpage
\section{System Model}
\labappendix{model}

We formulate our results for an asynchronous distributed system augmented with failure detectors \cite{CT96}.
This section recalls the fundamentals of this common model of computation then present some technical lemmas.
These lemmas are used in the follow-up to establish our complexity results regarding Leaderless SMR.

\subsection{Model}
\labappendix{model:model}

We consider an asynchronous distributed system consisting of a finite set of \emph{processes} $\procSet = \{p_1, p_2, \ldots, p_n\}$.
Processes may fail-stop, or \emph{crash}, and halt their computations.
A failure pattern is a function $F : \naturalSet \rightarrow 2^{\procSet}$ that captures how processes crash over time.
Processes that crash never recover from crashes, that is, for all time $t$, $F(t) \subseteq F(t+1)$.
If a process fails, we shall say that it is \emph{faulty}.
Otherwise, if the process never fails, it is said \emph{correct}.

\paragraph{Failure detectors}
A failure detector is an oracle $\fd$ that processes may query locally during an execution.
This oracle abstracts information, regarding synchrony and failures, available to the processes.
More precisely, a failure detector $\fd$ is a mapping that assign to a failure pattern $F$, one or more histories $\fd(F)$.
Each history $H \in \fd(F)$ defines for each process $p$ in the system, the local information $H(p,t)$ obtained by querying $\fd$ at time $t$.
The co-domain of $H : \procSet \times \naturalSet \rightarrow R$ is named the range of the failure detector.
An environment, denoted $\E$, is a set of failure patterns.

In the vein in \cite{realistic}, we only consider \emph{realistic} failure detectors.
This class of failure detectors cannot forecast the future.
This means that if two failure patterns $F$ and $F'$ are identical up to time $t$, then for any history $H \in \fd(F)$, there exists $H' \in \fd(F')$ identical up to time $t$ to $H$.

\paragraph{Message buffer}
Processes communicate with the help of messages taken from some set $\mathit{Msg}$.
A message $m$ is sent by some sender ($\sender(m)$) and addressed to some recipient ($\dst(m)$).
The sender may define some content ($\payload(m)$) before sending the message.
A message buffer, denoted $\BUFF$ , contains all the messages that were sent but not yet received.
More precisely, $\BUFF$ is a mapping from $\procSet$ to $2^{\mathit{Msg}}$.
When a process $p$ attempts to receive a message, it either removes some message from $\BUFF[p]$, or returns a special null message.
Note that $p$ may receive the null message even if the message buffer contains a message $m$ addressed to $p$.

\paragraph{Protocol}
A protocol $\mathcal{P}$ consists of a family of $n$ deterministic automata, one per process in $\procSet$.
Computation proceeds in steps of these automata.
At each step, a process $p$ executes atomically one of the following instructions:
\begin{inparaenumorig}[]
\item receive some message $m$;
\item fetch some value $d$ from the local failure detector module;
\item change its local state according to $\mathcal{P}$; or
\item send some message $m$ to another process.
\end{inparaenumorig}%
A \emph{configuration} of algorithm $\mathcal{P}$ specifies the local state of each process as well as the messages in transit (variable $\BUFF$).
In some initial configuration of $\mathcal{P}$, no message is in transit and each process $p$ is in some initial state as defined by $\mathcal{P}$.

\paragraph{Runs}
A run of algorithm $\mathcal{P}$ using failure detector $D$ in environment $\E$ is a tuple $\run = (F,H,I,S,T)$ where
\begin{inparaenumorig}[]
\item $F$ is a failure pattern in $\E$,
\item $H$ is a failure detector history in $D(F)$,
\item $I$ is an initial configuration of $\mathcal{P}$,
\item $S$ is a sequence of steps of $\mathcal{P}$, and
\item $T$ is a growing sequence of times (intuitively, $T[i]$ is the time at which step $S[i]$ is taken).
\end{inparaenumorig}%
A run whose sequence of steps is finite (respectively, infinite) is called a finite (respectively, infinite) run.
Every run $\run$ must satisfy the following standard \emph(well-formedness) conditions:
\begin{inparaenumorig}[\em(i)]
\item No process take steps after crashing;
\item The sequences $S$ and $T$ are either both infinite, or they are both finite and have the same length; and
\item The sequence of steps $S$ taken in the run conforms to the algorithm $\mathcal{P}$, the timing $T$ and the failure detector history $H$.
\end{inparaenumorig}%
A run $\run$ is \emph{admissible} for $\mathcal{P}$, or simply \emph{is a run of} $\mathcal{P}$, when it is well-formed and in addition:
\begin{inparaenumorig}
\item[\emph{(fairness)}] If $\run$ is infinite, every correct process takes an infinite number of steps in $\run$.
\item[\emph{(reliable links)}] Every process that infinitely often retrieves a message from $\BUFF$ eventually receives every message addressed to it.
\end{inparaenumorig}%
We shall write  $\runs^{\mathcal{P}}$ the runs of algorithm $\mathcal{P}$.
The superscript is ommitted when the algorithm we refer to is unambiguous.

Our results mostly concern nice runs \cite{niceRun}, that is failure-free runs during which the failure detector behave ``perfectly''.
More specifically, we consider that a run is \emph{nice} when there is no failure and the failure detector returns a constant value to the local process.
$\nruns^{\mathcal{P}}$ denote the nice runs of algorithm $\mathcal{P}$.

\paragraph{Additional notations}
When the context is clear, we do not distinguish a run from its sequence of steps.%
\footnote{
  This is particularly true for a nice run, since the failure detector history is constant.
}
Below, we introduce a handful of operators and shorthands that leverage this simplification.

Consider two sequence of steps $\lambda$ and $\lambda'$.
We note $(\lambda|P)$ the sub-sequence of steps taken by the processes $P \subseteq \procSet$ in $\lambda$.
Function $\proc(\lambda)$ returns the processes that take steps in $\lambda$.
We say that $\lambda$ is \emph{indistinguishable} from $\lambda'$ to $P \subseteq \procSet$ when $\lambda|P=\lambda'|P$.
As usual, if $\lambda \in \runs$, $\lambda'$ is indistinguishable from $\lambda$ to $\procSet$ and $\lambda'$ is well-formed, then $\lambda' \in \runs$.
If $\lambda = \lambda[0] \ldots s \ldots \lambda[n]$, then $(\lambda | \leq s)$ is the sequence $\lambda[0] \ldots s$, and $(\lambda | \geq s)$ equals $s \ldots \lambda[n]$.
The empty sequence is written $\epsilon$.
We note $\pref$ and $\suff$ respectively the prefix and suffix relations over the set of sequences.

Assume that an operation $\op$ is invoked in $\lambda$ then later returns some response $r$ to the local process.
This corresponds respectively to the steps $\inv(op)$ and $\resp(op,r)$ in $\lambda$.
We note $(\lambda | \op)$ the sub-sequence of steps ($\lambda | \haeq \inv(\op) | \hbeq \resp(\op,r)$), where $\hbeq$ and $\haeq$ are respectively the reflexive closure of the happen-before relation ($\hb$) and the reflexive closure of the converse of $\hb$.
For instance, ($\lambda | \announce(c)$) refers to the steps taken in $\lambda$ to announce command $c$.

\subsection{Technical Lemmas}
\labappendix{model:lemmas}

Below, we state a few results that concern nice runs of a protocol.
As pointed out previously, in a nice run there is no failure and the failure detector always returns the same value.
In this context, our first lemma is similar to Lemma~1 in FLP \cite{flp}.

\begin{lemma}
  \lablem{tools:flp}
  Consider two finite nice runs $\lambda \lambda'$ and $\lambda \lambda''$.
  If $\proc(\lambda') \inter proc(\lambda'') = \emptyset$, then $\lambda \lambda' \lambda''$ is a nice run.
\end{lemma}

\begin{proof}
  Follow from the model definition.
\end{proof}

\begin{lemma}
  \lablem{tools:1}
  Consider a nice run $\lambda = \lambda' s s' \lambda''$.
  If $\proc(s) \neq \proc(s')$ and there is no message $m$ such that $s = \send(m)$ and $s' = \recv(m)$, then $\hat{\lambda} = \lambda' s' s \lambda''$ is a nice run.
\end{lemma}

\begin{proof}
  First, we establish the well-formedness of $\hat{\lambda}$.
  Consider some receive step $recv(m)$ in $\hat{\lambda}$.
  As $\lambda$ is well-formed, $\send(m) <_{\lambda} \recv(m)$.
  By hypothesis, $s \neq \send(m)$, thus $\send(m) <_{\hat{\lambda}} \recv(m)$.
  
  Then choose some process $p$.
  We have, $(\hat{\lambda}|p) = (\lambda' | p) (s' s | p) (\lambda'' | p)$, by distributivity of the projection operator.
  It remains to show that $(s' s | p) = (s s' | p) $.
  There are three cases to consider:
  (Case $\proc(s) = p$). As $\proc(s) \neq \proc(s')$, we have $(s' s | p) = s = (s s' | p) $.
  (Case $\proc(s') = p$). This case is symmetrical to the previous one.
  (Otherwise). We have $(s' s | p) = \epsilon = (s s' | p) $.
  It follows that $\hat{\lambda}$ is indistinguishable from $\lambda$ to $\procSet$.

  Since $\hat{\lambda}$ is well-formed and $\hat{\lambda}$ is indistinguishable from $\lambda$ to $\procSet$, then $\hat{\lambda}$ is a run.
  This run has the same failure pattern as $\lambda$, i.e., it is failure-free.
  Moreover, the failure detector behave perfectly.
  As a consequence, $\hat{\lambda} \in \nruns$.
\end{proof}

In the above lemma $s'$ left-move with $s$ \cite{moverness}, written $s \lmove s'$.
By extension, we may deduce that for some run $\lambda \lambda'$ and some set of processes $P$,

\begin{corollary}
  \labcor{tools:1}
  If none of the messages received in $S=(\lambda'|P)$ was sent in $(\lambda' \setminus S)$, then $\lambda S(\lambda' \setminus S)$ is a nice run
\end{corollary}

\begin{corollary}
  \labcor{tools:2}
  If none of the messages sent in $S=(\lambda'|P)$ is received in $\lambda'$, then $\lambda (\lambda' \setminus S)S$ is a nice run.
\end{corollary}

When $A$ and $B$ are sequences of steps, $A \lmove B$ denotes that $B$ left-moves with $A$, that is $\forall (s,s') \in A \times B \ldotp s \lmove s'$.

\begin{lemma}
  \lablem{tools:4}
  $\forall \lambda \in \nruns \ldotp \forall \lambda' \pref \lambda \ldotp \lambda' \in \nruns$
\end{lemma}

\begin{proof}
  Choose some run $\lambda$ and some prefix $\lambda' \pref \lambda$.
  By construction $\lambda'$ is a run.
  Moreover as the failure detector behave perfectly in $\lambda$, it also behaves perfectly in $\lambda'$.
  From which it follows that $\lambda'$ is a run.
\end{proof}

\section{Proofs of \reftheotwo{reduction}{deciding}}
\labappendix{proofs}


This section contains the proofs of the theorems stated in \refsection{leaderless} which we deferred for readability.

\reduction*

\begin{proof}
  We build a Generic SMR protocol $A$ atop a Leaderless SMR protocol $B$ as follows.
  Each node running protocol $B$ holds a local copy of a partially ordered log $L$.
  This log is initially empty.
  Operation $\submit(c)$ in $A$ is mapped to operation $B.\submit(c)$.
  When command $c$ gets executed in $B$, we apply the following update to $L$: $L \leftarrow L \append \: c$.
  Clearly, this construction maintains that $L$ is a partially ordered graph over time.
  Furthermore, at the light of the definition of the operator $\append$, it is easy to see that any two conflicting commands gets ordered.
  In addition, this construction satisfies the three properties that define Generic SMR, as shown below.
  \begin{description}
  \item[\emph{(Non-triviality)}]
    In algorithm $B$, a command $c$ appears in some process dependency graph only if it was submitted before.
    Hence, in algorithm $A$, $c$ is in $L$ only if it was submitted before.
  \item[(Stability)]
    Recall that this property holds when, for any partially ordered log $L$, at any point in time $t$, $L_t \sqsubseteq L_{t+1}$.
    The $\append$ operator does not remove edges or nodes, thus $L_t$ is a subgraph of $L_{t+1}$.
    Now assume, for the sake of contradiction, that $L_t$ does not prefix $L_{t+1}$.
    There must exist $\cmdc$ in $L_t.V$ such that $(\cmdd, \cmdc) \in L_{t+1}.E$ and $(\cmdd, \cmdc) \notin L_t.E$.
    If $(\cmdd,\cmdc) \notin L_t$ then $\cmdd \in L_t.V$, as $\cmdd$ cannot be added between time $t$ and $t+1$ by definition of $\append$.
    This leads to $(\cmdd, \cmdc) \notin L_{t+1}$;
    a contradiction.
  \item[\emph{(Consistency)}]
    We prove that for any two processes $p$ and $p'$, the set $\{L_{p}, L_{p'}\}$ is compatible.
    To achieve this, we show that $L = (L_{p} \union L_{p'})$ is a partially ordered log that suffixes both $L_{p}$ and $L_{p'}$.
    To this end, let us consider two conflicting commands $\cmdc$ and $\cmdd$ in $L$.

    \begin{claim}
      \labclaim{reduction:1}
      Commands $\cmdc$ and $\cmdd$ cannot be in $L_{p}.V \xor L_{q}.V$.
    \end{claim}

    \begin{proof}
    By contradiction, assume $\cmdc$ and $\cmdd$ belong to different logs (wlog. say respectively $L_{p}$ and $L_{p'}$).
    Applying Invariant~1 to $p$ leads to the fact that $d \notin \deps^*(\cmdc)$. 
    Symmetrically from process $p'$, we have that $\cmdc \notin \deps^*(\cmdd)$.
    A contradiction to the Consistency property of Leaderless SMR. 
    \end{proof}

    \begin{claim}
      \labclaim{reduction:2}
      For any process $q \in \{p,p'\}$, if $(\cmdc,\cmdd)$ is in $L.E$ and $\cmdd$ is in $L_{q}.V$ then $(\cmdc,\cmdd)$ is in $L_{q}.E$.
    \end{claim}

    \begin{proof}
      Since $(\cmdc,\cmdd) \in L.E$,  $(\cmdc,\cmdd)$ belongs to (say) $L_{p}.E$. 
      This leads to $\cmdc \eorder \cmdd$ at $p$
      By Invariants 1 and 2, $\cmdc \in \deps^*(\cmdd)$ at $p$.
      Now, if $\cmdd \in L_{q}.V$, $q$ executes $\cmdd$.
      It follows that $\cmdd$ was stable at $q$.
      By the Stability property of Leaderless SMR, $\cmdc \in \deps^*(\cmdd)$ at $q$.
      As a consequence, $\cmdc \eorder \cmdd$ at $q$ and $(\cmdc,\cmdd)$ is in $L_{q}.E$.
    \end{proof}
      
    The end of the proof goes as follows.
    \begin{enumerate}
    \item $L_{p} \pref L \land L_{p'} \pref L$ \\
      From \refclaim{reduction:2}.
    \item $L$ is a partially ordered log.
      \begin{enumerate}
      \item $\forall \cmdc,\cmdd \in L.V \ldotp \cmdc \conflict \cmdd \implies ((\cmdc,\cmdd) \in L.E \lor (\cmdd,\cmdc) \in L.E)$ \\
        From \refclaim{reduction:1}.
      \item $L$ is a directed acyclic graph.\\
        If $L$ is cyclic, by \refclaim{reduction:2}, either $L_{p}$ or $L_{p'}$ is cyclic.
        Contradiction.
      \end{enumerate}      
    \end{enumerate}
  \end{description}
\end{proof}

\deciding*

\begin{proof}
  Assume a run $\lambda$ of \refalg{deciding}.  
  \begin{description}
    \item [\emph{(Validity)}]
      Let $\cmdc$ be a decided command at some process $p$ in $\lambda$.
      By definition, $\deps(\cmdc) \neq \bot$ at process $p$.
      This is only possible through the execution of \refline{commit:5} or \refline{commit:8}.
      If \refline{commit:5} gets executed, then $p$ took step $\submit(\cmdc)$ previously.
      Otherwise, \refline{commit:8} is executed.
      To execute \refline{commit:8} it is necessary to receive a message $(\cmdc,D)$ at \refline{commit:7}.
      Such a message is sent at \refline{commit:6} by some process $q$.
      Executing \refline{commit:6} implies that \refline{commit:5} was executed by $q$ before, which boils down to the previously analyzed case.

    \item [\emph{(Consistency)}]
      Let $\cmda$ and $\cmdb$ be two conflicting committed commands at some process $p$.
      By definition, the two commands are committed when $\deps(\cmda), \deps(\cmdb) \in 2^{\cmdSet}$ holds at $p$.
      Similarly to the Validity property, this is the case when \refline{commit:5} or \refline{commit:8} is executed.
      As seen in the proof of the validity property, executing \refline{commit:8} can be traced back to the execution of \refline{commit:5}.
      Wlog. we only analyze the case where \refline{commit:5} is executed hereafter.
      
      By the Validity property of consensus, the value $D$ assigned to $\deps(\cmda)$ at $p$ was proposed before by a process $q$.
      The operation that yields $D$ is executed at \refline{commit:3} by $q$ through the use of the $\dds$ service.
      Similarly, we may define a process $q'$ such that the value of $D'=\deps(\cmdb)$ at $p$ is the result of a call to $\announce(\cmdb)$ by $q'$.

      By the Visibility property of the $\dds$ service, we have either $\cmdb \in D$ or $\cmda \in D'$.
      Hence, the Consistency property of Leaderless SMR holds.

    \item [\emph{(Stability)}]
      Assume that $\cmdc$ is eventually stable at processes $p$ and $p'$. 
      Let $E$ and $E'$ be the value of $\deps(\cmdc)$ at respectively $p$ and $p'$ when this happens.
      In what follows, we prove that $E'=E$.
     
      To establish this result, we first show that processes agree over aborted commands.
      Let $q$ and $q'$ be two processes that decide some command $\cmda$.
      As with prior properties, we assume wlog. that such a decision is taken at \refline{commit:5}.
      Assume $q$ aborts $\cmda$.
      The output of \refline{commit:5} depends on the computation at \reflinestwo{commit:3}{commit:4}.
      Since command $\cmda$ is aborted, necessarily $q$ takes the slow path.
      By the Validity property of $\cons_{\cmda}$, $\top$ is proposed to $\cons_{\cmda}$ by some process $q''$.
      This value is the response of $\announce(\cmda)$ at \refline{commit:2} by $q''$.
      Then, consider the following two cases.
      \begin{inparaenum}
      \item If $q'$ takes also the slow path, by the Agreement property of $\cons_{\cmda}$, it should also abort $\cmda$.
      \item Otherwise $q'$ follows the fast path.
        In that case, by the Validity property of $\cons_{\cmda}$, $q'$ returns some value $(\any,\true)$ from $\announce(\cmda)$ at \refline{commit:2}.
        Since $q''$ returns $\top$ from $\announce(\cmda)$, this case contradicts the Weak Agreement property of the $\dds$ service.
      \end{inparaenum}

      Now, let $F$ and $F'$ be the value of $\deps(\cmdc)$ first assigned by respectively $p$ and $p'$.
      Wlog. assume that this assignment occurs at \refline{commit:5}.
      We show that for every command $\cmda \in F \xor F'$, $\cmda$, calling $\announce(\cmda)$ returns $\top$.
      If $p$ and $p'$ take the slow path, then $F=F'$ by the Agreement property of $\cons_{\cmdc}$.
      Otherwise, one of the two processes, say $p$, takes the fast.
      This implies that $p$ returns $(F,\true)$ from $\announce(\cmda)$.
      If $p'$ takes also the fast path, $(F',\true)$ is returned from $\announce(\cmda)$.
      Otherwise, by the Validity property of $\cons_{\cmda}$, some process $p''$ returns $(F',\false)$ from $\announce(\cmda)$.
      In both cases, the claim follows from the Weak Agreement property of the $\dds$ service.

      We now prove that $F$ and $F'$ converge toward $E$, the value of $\deps(\cmdc)$ when $\cmdc$ is stable at $p$.
      These sets are bounded and the only update operation is the removal of an aborted command.
      Pick $\cmda \in F \xor F'$.
      \begin{inparaenumorig}
      \item[(Case $\cmda \in F \setminus F'$)]
        Since $\cmdc$ is stable, eventually $\cmda$ is decided.
        Thus, eventually $p$ removes $\cmda$ from $F$.
      \item[(Otherwise)]
        Symmetrical to the previous one.
      \end{inparaenumorig}
  \end{description}
\end{proof}

\section{The ROLL Theorem}
\labappendix{roll}

Below, we define formally the three properties introduced in \refsection{leaderless:properties}.
Then, we present two technical lemmas that characterize the behavior of Leaderless SMR protocols during nice runs.
These lemmas are used to show the ROLL theorem and the chaining effect in the next section.

\subsection{The Properties}
\labappendix{roll:prop}

In \refsection{leaderless:properties}, we introduce Reliability, Optimal Latency and Load Balancing as three core properties of Leaderless SMR.
Reliability guarantees that the protocol makes progress even if up to $f$ failures occur.
This means that once a command is submitted, the protocol must take a decision, possibly aborting it.
Given a run $\lambda$, let us write $\cmdc \in \lambda$ when $\submit(\cmdc)$ is invoked in $\lambda$.
Then, we define this property as follows.

\begin{description}
\item[(\emph{Reliability})]
  In every run, if there are at most $f$ failures, every submitted command gets eventually decided at every correct process.
  \begin{displaymath}
    \forall \lambda \in \uruns \ldotp \forall \cmdc \in \lambda \ldotp \forall q \in \correct(\lambda) \ldotp \faults(\lambda) \leq f  \implies decide_q(c) \in \lambda
  \end{displaymath}
\end{description}

Optimal Latency requires that in a nice run every command commits after two message delays.
Moreover, in the absence of contention the command is immediately stable.

Message delays measure the time complexity of a sequence of steps, neglecting the cost of local computations.
In detail, the latency of a causal path $\rho$, written $\latency(\rho)$, is the number of consecutive $\send(m)$ then $\recv(m)$ steps in $\rho$.
Denoting $\hb$ the happens-before relation \cite{lamportTime} in a run $\lambda$, $\cpaths(\lambda)$ contains all the maximal chains in $(\lambda, \hb)$.
The latency of $\lambda$ is then defined as $\Delta(\lambda) = max\{\latency(\rho): \rho \in cpaths(\lambda)\}$.

To track contention during a run, we introduce function $\contended(\lambda,\cmdc)$.
This function returns $\true$ when there exists a command $\cmdd$ conflicting with $\cmdc$ such that
\begin{inparaenumorig}[]
\item $\cmdd$ is submitted before $\cmdc$, and
\item $\cmdd$ is not committed at $\coordc$ when $\cmdc$ is submitted.
\end{inparaenumorig}
Optimal Latency is then specified as follows.

\begin{description}
\item[(\emph{Optimal Latency})]
  During a nice run, every call to $\announce(d)$ returns a tuple $(D,\flag)$ after two message delays such that
  \begin{inparaenum}
  \item if there is no concurrent conflicting command to $c$, then $\flag$ is set to $\true$,
  \item $D \in 2^{\cmdSet}$, and
  \item for every $d \in D$, $d$ was announced before.
  \end{inparaenum}
  \begin{displaymath}
    \begin{array}{l@{~}l@{~}l}
    \forall \lambda \in \nruns \ldotp \forall \cmdc \in \cmdSet \ldotp 
    \resp(\announce(\cmdc),(D,\flag)) \in \lambda
    \implies
    & \land & \latency(\lambda|\announce(\cmdc)) = 2 \\
    & \land & (\flag = \true \lor \contended(\lambda,\cmdc)) \\
    & \land & D \in 2^{\cmdSet} \\
    & \land & \forall d \in D \ldotp \invocation{p}{\announce(\cmdd)} \\
    &       & \hspace{5em} \hb_{\lambda} \responseVoid{q}{\announce(\cmdc)}
    \end{array}
  \end{displaymath}
\end{description}

An important property of Leaderless SMR protocols is to distribute the task of ordering conflicting commands across processes.
This characteristic is captured by the Load Balancing property.
In detail, this property requires that during a nice run
\begin{inparaenum}
\item progress can be made using any fast path quorum, and
\item when returning from the announcement of a command, no message gets undelivered.
\end{inparaenum}

\begin{description}
\item[(\emph{Load balancing})]
  During a nice run, any fast path quorum in $\fquorums(\cmdc)$ can be used to announce a command $\cmdc$.
  \begin{displaymath}
    \begin{array}{l@{~}l@{~}l}
      \forall \lambda \in \nruns \ldotp \forall \cmdc \notin \lambda \ldotp \forall Q \in \fquorums(\cmdc) \ldotp \exists \lambda' \ldotp
      & \land & \lambda(\lambda'|\announce(\cmdc)) \in \nruns \\ 
      & \land & \proc(\lambda'| \announce(\cmdc)) = Q \\ 
      & \land & \pending(\lambda'|\announce(\cmdc)) = \emptySet 
    \end{array}
  \end{displaymath}
  %
\end{description}

\subsection{Characterizing ROLL protocols}
\labappendix{characterizing}


Assume a ROLL protocol $\mathcal{P}$ and consider some nice run $\lambda$ of $\mathcal{P}$ during which command $\cmdc$ is announced.
Since $\mathcal{P}$ ensures Optimal Latency, announcing $\cmdc$ takes two message delays.
This means that during the announcement of $\cmdc$ a set of requests is sent by the coordinator to which a set of processes replies.
The lemma below characterizes precisely this pattern, where $\Gamma$ is a shorthand for $\lambda|\announce(\cmdc)$.

\begin{lemma}
  \lablem{roll:0}
  $
  \forall \mathscr{C} \in \cpaths(\Gamma) \ldotp \exists m, m', \rho, \rho', \rho'' \ldotp
  \mathscr{C}
  =
  \invocation{\coordc}{\announce(\cmdc)}
  \cdot
  \rho
  \cdot
  \send(m)
  \cdot \\
  \recv(m)
  \cdot
  \rho'
  \cdot
  \send(m')
  \cdot
  \recv(m')
  \cdot
  \rho''
  \cdot
  \responseVoid{\coordc}{\announce(\cmdc)}
  $
\end{lemma}

\begin{proof}
  Let $s$ be the smallest element in $\mathscr{C}$.
  Since $\mathscr{C} \in \cpaths(\Gamma)$ holds, we have that $s \in \lambda |\haeq \invocation{\coordc}{\announce(\cmdc)}$.
  Therefore $\invocation{\coordc}{\announce(c)} \hbeq s$.
  The causal path $\mathscr{C}$ is a maximal chain in $\Gamma$, thus $s = \invocation{\coordc}{\announce(\cmdc)}$.
  Analogously, if we define $s'$ as the largest element in $\mathscr{C}$, we have that $s' = \responseVoid{\coordc}{\announce(\cmdc)}$.

  Then, $\latency(\mathscr{C})=2$ by Optimal Latency.
  As a consequence, there exists two messages $m$ and $m'$ and three sequences of steps $\rho$, $\rho'$ and $\rho''$ such that
  \begin{inparaenum}
  \item $m$ is send by $\coordc$ after $\invocation{\coordc}{\announce(\cmdc)} \cdot \rho$;
  \item $m$ is received by some process $q$ that executes the steps $\rho'$ then sends a message $m'$ to $\coordc$; and
  \item $\coordc$ receives $m'$ and executes the steps $\rho''$ before the step $\responseVoid{\coordc}{\announce(\cmdc)}$.
  \end{inparaenum}
\end{proof}

At the light of the above characterization, we call \emph{request} the first message exchanged in some causal path of $\Gamma$.
Similarly, a reply is the answer received by $\coordc$.
We note $m_{\cmdc}$ and $m'_{\cmdc}$ respectively the last request sent and the first reply processed by the coordinator.
The steps $S_{\cmdc}$, $M_{\cmdc}$ and $R_{\cmdc}$ below define a partitioning of $\Gamma$.

\begin{itemize}
\item $S_{\cmdc}$ are the steps taken from announcing $c$ to the sending of $m_{\cmdc}$ at the coordinator.
    Formally, $S_{\cmdc} = (\Gamma|\coord(c)|\leq \send(m_{\cmdc}))$.
\item $R_{\cmdc}$ are the steps taken by $\coord(c)$ after receiving $m'_{\cmdc}$ until the announcement returns.
  In other words, $R_{\cmdc} = (\Gamma|\coord(c)|\geq \recv(n_{\cmdc}))$.
\item $M_{\cmdc}$ are the steps taken during the announcement of $c$ which are neither in $S_{\cmdc}$, nor in $R_{\cmdc}$.
  That is, $M_{\cmdc} = (\Gamma \setminus (S_{\cmdc} \union R_{\cmdc}))$.
\end{itemize}

Based upon this partitioning, the two lemmas that follow characterize the behavior of ROLL protocols during nice runs.
They are the basic building blocks of our two complexity results.

\begin{lemma}
  \lablem{roll:1}
  $\forall \lambda \in \nruns \ldotp \forall \cmdc \notin \lambda \ldotp \lambda S_{\cmdc} M_{\cmdc} R_{\cmdc} \in \nruns \land \latency(M_{\cmdc}) = 0$.
\end{lemma}

\begin{proof}
  Assume a command $\cmdc \notin \lambda$ and some fast path quorum $Q \in \fquorums(\cmdc)$.
  By Load Balancing, we may suffix $\lambda$ with $\lambda'$ such that
  \begin{inparaenumorig}[]
  \item $\announce(\cmdc)$ occurs in $\lambda'$,
  \item only the processes in $Q$ execute steps in $\lambda'|\announce(\cmdc)$,
  \item and $\lambda(\lambda'|announce(c))$ is a nice run.
  \end{inparaenumorig}
  From the definitions of $S_{\cmdc}$, $M_{\cmdc}$ and $R_{\cmdc}$, $\lambda'$ and $S_{\cmdc} M_{\cmdc} R_{\cmdc}$ contain the same steps.
  It remains to show that taking such steps in this order remains legal.
  
  Let us note $\Gamma =\lambda'|announce(c)$.
  Observe that none of the messages received in $S_{\cmdc}$ was sent in $\Gamma \setminus S_{\cmdc}$.
  Otherwise, by Optimal Latency, $\cmdc$ takes (at least) three message delays.
  Applying \refcor{tools:1}, $\lambda S_{\cmdc} (\Gamma \setminus S_{\cmdc}) \in \nruns$.
  In the same vein, none of the messages sent in $R_{\cmdc}$ is received in $M_{\cmdc} = (\Gamma \setminus S_{\cmdc}) \setminus R_{\cmdc}$.
  Therefore we have that $\lambda S_{\cmdc}((\Gamma \setminus S_{\cmdc})\setminus R_{\cmdc}) R_{\cmdc} \in \nruns$ by \refcor{tools:2}.
  By definition, $M_{\cmdc} = (\Gamma \setminus S_{\cmdc}) \setminus R_{\cmdc}$.

  Now, assume by contradiction that $\latency(M_{\cmdc}) > 0$.
  Let $m$ be the message exchanged during $M_{\cmdc}$.
  Since $\send(m)$ and $\recv(m)$ belongs to $\Gamma$, there exists a causal path $\mathscr(C)$ starting with $\invocation{\coordc}{\announce(\cmdc)}$ and ending  with $\response{\coordc}{\announce(\cmdc)}$ that contains these two steps.
  Applying \reflem{roll:0}, $\mathscr{C}$ is of the form $\invocation{\coordc}{\announce(\cmdc)} \cdot \rho \cdot \send(m) \cdot \recv(m) \cdot \rho' \cdot \send(m') \cdot \recv(m') \cdot \rho'' \cdot \responseVoid{\coordc}{\announce(\cmdc)}$.
  Message $m$ is neither a request, nor a reply, this implies that $\latency(\rho')=1$, and thus $\latency(\mathscr{C}) \geq 3$, contradicting Optimal Latency.
\end{proof}

\begin{lemma}
  \lablem{roll:2}
  $\forall \lambda \in \nruns \ldotp \forall P \subseteq \procSet \ldotp \lambda S_{\cmdc}(M_{\cmdc}|P) \in \nruns$
\end{lemma}
\begin{proof}
  We start by picking $c\in \cmdSet$ such that $c \notin \lambda$. This allows us to apply \reflem{roll:1} to achieve the following: $\lambda S_{\cmdc} M_{\cmdc} R_{\cmdc} \in \nruns$.
  Applying \reflem{tools:4}, we have that: $\lambda S_{\cmdc} M_{\cmdc} \in \nruns$.
  Then, let $X = M_{\cmdc}|P$, we know then that none of the messages received in $X$ were sent in $M_{\cmdc} \setminus X$ (by \reflem{roll:1}, $\latency(M_{\cmdc}) = 0$).
  Therefore, we can apply \refcor{tools:1} to obtain $\lambda S_{\cmdc} X(M_{\cmdc} \setminus X) \in \nruns$.
  Finally, applying \reflem{tools:4} gives us: $\lambda S_{\cmdc} (M_{\cmdc}|P) \in \nruns$. 
\end{proof}

\subsection{Proof}
\labappendix{roll:proof}

We now proceed to proving the ROLL theorem with the above two technical lemmas.
The proof follows the sketch depicted in \refsection{roll}.

\roll*

\begin{proof}
  By contradiction, assume a ROLL protocol that satisfies $2F+f-1 \leq n$.
  Let $\cmd{1}$ and $\cmd{2}$ be two conflicting commands in $\cmdSet$.
  
  Define a partitioning $P_1$, $P_2$ and $Q$ of $\procSet$, the set of processes, such that
  \begin{inparaenum}
  \item $P_1 \inter P_2 = \emptySet$;
  \item $Q = \procSet (\setminus P_1 \union P_2)$;
  \item $\cardinalOf{P_1} = \cardinalOf{P_2} = F-1$; and
  \item $|Q| = n-2(F-1)$.
  \end{inparaenum}

  The CAP impossibility result \cite{cap} tells us that $2F < n$.
  As a consequence, there exist at least two distinct processes $p_1$ and $p_2$ in $Q$.
  Let $Q_1 = P_1 \inter Q \setminus\{p_2\}$ and $Q_2 = P_2 \union Q \setminus \{p_1\}$.

  By applying \reflem{roll:1} to the initial state, $\run_1 = S_{1} M_{1} R_{1} \in \nruns$, 
  By Optimal Latency, the fast path is taken in $\run_1$ and $\deps(\cmd{1})=\emptySet$ holds at $p_1$ at the end of the run.
  Symmetrically, we define $\run_2 = S_{2} M_{2} R_{2} \in \nruns$.

  By applying \reflem{roll:2} to $\run_1$, $\lambda' = S_{1} (M_{1}|P_1) \in \nruns$.
  In a symmetrical manner, $\lambda'' = S_{2} (M_{2}|P_2) \in \nruns$.  
  Observe that $proc(\lambda')\inter \proc(\lambda'') = \emptySet$.
  Thus, by \reflem{tools:flp}, we obtain $S_{1} (M_{1}|P_1) S_{2} (M_{2}|P_2) \in \nruns$.
  From \refcor{tools:1}, $\sigma = S_{1} S_{2} (M_{2}|P_2) (M_{1}|P_1) \in \nruns$.

  Let $\hat{\sigma}$ be an infinite suffix of $\sigma$ in which all the processes in $Q$ are faulty.
  Such a suffix exist because we consider realistic failure detectors.
  As $\cardinalOf{Q} \leq f$, $\cmd{1}$ and $\cmd{2}$ are eventually stable by Reliability.
  Let $\sigma'$ be the shortest suffix of $\hat{\sigma}$ for which this is true at some process $p$.
  We define $\run_3 = \sigma \sigma'$.
  By construction, $\proc(\sigma') \subseteq P_1 \union P_2$.

  In what follows, we construct a fourth run, $\run_4$.
  We will show that $\run_4$ is indistinguishable from $\run_3$ to $p_2$, while at the same time being indistinguishable from $\run_1$ to $p_1$.
  This implies that the same decision about $\deps(\cmd{1})$ at $p_1$ in $\run_1$ is taken by $p$ in $\run_3$.

  By \reflem{roll:1} applied to $\lambda''$, $S_{2}(M_{2}|P_2) S_{1}M_{1}R_{1} \in \nruns$.
  Then, partitioning $M_{1}$ using \reflem{roll:2}, $S_{2} (M_{2}|P_2) S_{1} (M_{1}|P_1) (M_{1}|Q^* \union \{p_1\}) R_{1} \in \nruns$.
  By applying \refcor{tools:1}, we obtain $\sigma (M_{1}|Q^* \union \{p_1\}) R_{1} \in \nruns$.
  
  Let $\run_4$ be a run with the same failure pattern and failure detector history as $\run_3$, and in which the steps $\sigma (M_{1}|Q^* \union \{p_1\}) R_{1} \sigma'$ are taken.
  Since $Q \inter (P_1 \union P_2) = \emptySet$, $\run_4$ is not distinguishable from $\sigma (M_{1}|Q^* \union \{p_1\}) R_{1}$ to $Q$.
  Similarly, $\run_4$ is not distinguishable from $\sigma \sigma'$ to $P_1 \union P_2$.
  Thus $\run_4$ is well-formed and a run of the protocol.
  
  \begin{claim}
    \labclaim{roll:1}
    $\run_1|\{p_1\} = \run_4|\{p_1\}$
  \end{claim}

  \begin{proof}
    $\run_1|p_1 = (S_1 M_1 R_1) | \{p_1\} = S_1 (M_{1}| \{p_1\}) R_{1} = \run_4 | \{p_1\}$
  \end{proof}

  \begin{claim}
    \labclaim{roll:2}
    $\run_4|p = \run_3|p$
  \end{claim}

  \begin{proof}
    $\run_4|p = (M_2|p) (M_1|p) (\sigma'|p) = \run_3|p$
  \end{proof}

  \refclaim{roll:1} implies that $\cmd{1}$ is stable at $p_1$ in $\run_4$ and that $deps(\cmd{1}) = \emptySet$.
  Similarly, \refclaim{roll:2} leads to $\cmd{1}$ stable at $p$ in $\run_4$.
  In addition, the value of $deps(\cmd{1})$ at $p$ in $\run_4$ is the same as in $\run_3$.
  By the Stability property of Leaderless SMR, $deps(\cmd{1}) = \emptySet$ at $p$ in $\run_3$.

  A symmetric argument to the one above can be made using run $\run_2$ and a run $\run_5$.
  This leads to the conclusion that $p$ decides $deps(\cmd{2}) = \emptySet$ in $\run_3$.

  From what precedes, run $\run_3$ contradicts the Consistency property of Leaderless SMR.

\end{proof}

\section{Chaining effect}
\labappendix{chaining}

This section shows how a chaining effect may occur in ROLL-optimal protocols.
To establish this result, we construct a $\async{2}$ nice run with a pending chain of $n$ commands.
The run is built inductively starting from a solo run during which a single command is submitted.

At coarse grain, our construction works as follows.
Let $\sigma_i$ be a run with a chain of size $i$ and some pending command $\cmdi$.
We extend $\sigma_i$ with the partial announcement of a new command $\cmdin$.
Then, we take the decision for $\cmdi$ to obtain $\sigma_{i+1}$.
In $\sigma_{i+1}$, command $\cmdin$ is pending, yet $\coord(i)$ does not know if it is committed or not.
We argue that in the latter case, it could have missed $\cmdi$.
Thus, $\coord(i)$ must add $\cmdin$ to the dependencies of $\cmdi$, forming a new chain of size $i+1$.

\subsection{An Inductive Reasoning}
\labappendix{chaining:induction}

We build our result inductively using a family of $k$ distinct commands $(\cmdi)_{i \in [1,k]}$.
Each command is associated with a nice run ($\sigma_i$), a fast path quorum ($Q_i$), a subset of $f-1$ processes ($P_i$), and a process ($q_i$).
We shall establish that at rank $i$ the following property holds.

\begin{displaymath}
  \begin{array}{l@{~}l@{~}l}
    \mathfrak{P}(i) & \equaldef & \exists \sigma_i, Q_i, P_i, q_i, \Gamma_i, \Gamma'_i \ldotp \\
    && \land~ \sigma_{i} = \Gamma_i S_i (M_i|P_i) \Gamma'_i (M_i| Q_i \setminus P_i) R_i \\
    && \land~ \proc(\Gamma_i') \inter Q_i = P_i \\
    && \land~ (S_i (M_i | P_i) (\Gamma_i'| \procSet \setminus \{q_i\})) \lmove (\Gamma_i' | q_i) \\
    && \land~ \sigma_{i} \vdash \cmdi \cmdip \cdots \cmd{1} \\
    && \land~ \pending(\sigma_i) = \emptySet \\
    && \land~ \forall \rho \in \cpaths(\Gamma'_i) \ldotp \cardinalOf{\rho} \leq 1 \\
    && \land~ \sigma_i \in \async{2}
  \end{array}
\end{displaymath}

In the above definition, the first two clauses give the general form of $\sigma_i$.
The third clause indicates that the steps $(\Gamma_i' | q_i)$ left-move with $S_i (M_i | P_i)$.
As we shall see later, these steps are taken by the coordinator of the next command (i.e., $\cmdin$).
The fourth clause requires that a chaining effect occurs with prior commands.
The fifth and sixth clause upper-bound the asynchrony in $\sigma_i$.
They are used to show by induction that the run is \async{2} (last clause).

The remaining of this section is devoted to showing that $\mathfrak{P}(k)$ holds.
From which we may deduce the following theorem (in \refsection{chaining}).

\chaining*

\begin{proof}
  The formal statement of the theorem is:
  $
  \forall k>0 \ldotp \exists \lambda \in \async{2} \ldotp \exists (\cmdi)_{i \in [1,n]} \subseteq \cmdSet \ldotp
  \lambda \vdash \cmd{k} \cmd{k-1} \cdots \cmd{1} \land \cmd{n} \notin \sCommit
  $.
  To show this, let $\sigma_{k}$ be the $\async{2}$ run given by $\Gamma_k S_k (M_k|P_k) \Gamma'_k (M_k| Q_k \setminus P_k) R_k$, by applying $\mathfrak{P}(k)$.
  Consider its prefix $\Gamma_k S_k (M_k|P_k) \Gamma'_k$.
  Command $\cmd{k}$ is not committed in this run yet.
  Moreover, by $\mathfrak{P}(k)$, this run contains a chain $\cmd{k} \cmd{k-1} \cdots \cmd{1}$.
\end{proof}

\subsection{Preliminaries}
\labappendix{chaining:preliminaries}

For $i=1$, we chose $Q_1$ as any set of $(n-F)$ processes in $\procSet$.
$P_1$ is any subset of $f-1$ processes in $Q_1$.
This set is constructable since by ROLL-optimality $n-F-(f-1) = F-2$ and by CAP $F>1$.
$\coord(c_1)$ is any process in $Q_1 \setminus P_1$.
Process $q_1$ is chosen outside of $Q_1$.
Applying \reflem{roll:2} then \reflem{roll:1} leads to $S_1(M_1|P_1)(M_1|Q_1 \setminus P_1)R_1$.
This gives us immediately $\mathfrak{P}(0)$ where both $\Gamma_1$ and $\Gamma_1'$ are empty.

\reffigure{assump1} illustrates how we construct the quorum and process variables at rank $i+1$ from the ones at rank $i$.
Such a construction ensures the following list of facts that are used to establish $\mathfrak{P}(i>1)$.
\begin{enumerate}[label=({F\arabic*})]
\item \labfact{qf:1} $\coordin = q_i$
\item \labfact{qf:2} $P_{i+1} = Q_{i+1} \inter Q_{i}$
\item \labfact{qf:3} $P_{i+1} \inter P_i = \emptySet$
\item \labfact{qf:4} $q_{i+1} \notin Q_{i+1} \union \{ \coordi \}$.
\item \labfact{qf:5} $(\coordip \union P_{i+1}) \inter (\coordi \cup P_i) = \emptySet$
\end{enumerate}

In detail, we build $Q_{i+1}$ as $P_{i+1} \union \hat{P}_{i+1} \union \{\coordip\}$, where:
\begin{inparaenumorig}[(a)]
\item $\hat{P}_{i+1}$ are $n-F-f+1$ processes outside of $Q_i$;
\item $P_{i+1}$ are $f-1$ processes picked in $Q_i \setminus (P_i \union \coordi)$; and
\item $\coordin$ is set to $q_i$.
\end{inparaenumorig}%
Then, process $q_{i+1}$ is chosen in $Q_i$ outside of $P_{i+1} \union \{\coord{i}\}$.

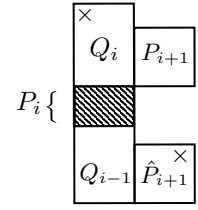
\begin{wrapfigure}{r}{0.3\textwidth}
  \centering

\tikzset{
pattern size/.store in=\mcSize, 
pattern size = 5pt,
pattern thickness/.store in=\mcThickness, 
pattern thickness = 0.3pt,
pattern radius/.store in=\mcRadius, 
pattern radius = 1pt}
\makeatletter
\pgfutil@ifundefined{pgf@pattern@name@_i1ytpxufr}{
\pgfdeclarepatternformonly[\mcThickness,\mcSize]{_i1ytpxufr}
{\pgfqpoint{0pt}{-\mcThickness}}
{\pgfpoint{\mcSize}{\mcSize}}
{\pgfpoint{\mcSize}{\mcSize}}
{
\pgfsetcolor{\tikz@pattern@color}
\pgfsetlinewidth{\mcThickness}
\pgfpathmoveto{\pgfqpoint{0pt}{\mcSize}}
\pgfpathlineto{\pgfpoint{\mcSize+\mcThickness}{-\mcThickness}}
\pgfusepath{stroke}
}}
\makeatother
\tikzset{every picture/.style={line width=0.75pt}} 

\begin{tikzpicture}[x=0.75pt,y=0.75pt,yscale=-1,xscale=1]

\draw   (80.35,150.72) -- (80.24,89.13) -- (110.29,89.07) -- (110.39,150.67) -- cycle ;
\draw   (80.6,189.47) -- (80.49,130.72) -- (110.54,130.67) -- (110.64,189.42) -- cycle ;
\draw   (110.54,130.67) -- (110.49,101.22) -- (140.54,101.16) -- (140.59,130.62) -- cycle ;
\draw  [pattern=_i1ytpxufr,pattern size=2.25pt,pattern thickness=0.75pt,pattern radius=0pt, pattern color={rgb, 255:red, 0; green, 0; blue, 0}] (80.78,150.72) -- (80.75,130.72) -- (110.28,130.67) -- (110.31,150.67) -- cycle ;
\draw   (110.64,189.42) -- (110.59,159.96) -- (140.64,159.91) -- (140.69,189.37) -- cycle ;

\draw (95.6,112) node    {$Q_{i}$};
\draw (96.47,175) node  [font=\small]  {$Q_{i-1}$};
\draw (126.17,114.67) node  [font=\small]  {$P_{i+1}$};
\draw (125.64,174.66) node  [font=\small]  {$\hat{P}_{i+1}$};
\draw (85.9,95) node  [font=\small]  {$\times $};
\draw (133.6,165.5) node  [font=\small]  {$\times $};
\draw (69.2,141) node    {$\big\{ $};
\draw (58,139) node    {$P_{i}$};

\end{tikzpicture}
  \caption{The symbol $\times$ represents coordinator placement}
  \labfigure{assump1}
\end{wrapfigure}

Let us establish the correctness of this construction.
For starter, defining $\hat{P}_{i+1}$ is possible.
Indeed, there are $F$ processes outside of $Q_i$ from which we need $n-F-f+1$ of them.
By ROLL-optimality, $F-(n-F-f+1) = -n+2F+f-1 \geq 0$.
Then, building $P_{i+1}$ requires $f-1$ processes in $Q_i \setminus (P_i \union \coordi)$.
We have $n-F=-1+f+F$ by ROLL and by CAP $F>1$.
Thus, there are enough processes.
The above reasoning also tells us that we may pick $q_{i+1}$ outside of $P_{i+1} \union \{\coord{i}\}$ in in $Q_i$.
Then,
\begin{inparaenumorig}[(a)]
\item \reffact{qf:1} is true by construction;
\item $Q_{i+1}$ is defined as $P_{i+1} \union \hat{P}_{i+1} \union \{\coordip\}$.
  $\hat{P}_{i+1}$ does not intersect with $Q_i$.
  At rank $i$, $q_i \notin Q_i$ and $\coordin = q_i$.
  Thus, $\coordin \notin Q_i$.
  This establishes \reffact{qf:2};
\item By construction, we have \reffact{qf:3}; and
\item $q_{i+1}$ is chosen in $Q_i$ outside of $P_{i+1} \union \{\coord{i}\}$.
  Hence, it is not in $\{Q_{i+1} \union \{\coordi\}$, giving us \reffact{qf:4}.
\item \reffact{qf:5} follows from the conjunction of \reffact{qf:3} and \reffact{qf:3}.
\end{inparaenumorig}

\subsection{Inductive step: $\mathfrak{P}(i) \implies \mathfrak{P}(i+1)$}
\labappendix{chaining:step}

\subsubsection{Construction}
\labappendix{chaining:step:construction}

Let us consider the prefix $\Gamma_i S_i (M_i | P_i) \Gamma_i'$ of $\sigma_i$.
Applying \reflem{roll:2}, we obtain the following nice run: $\Gamma_i S_i (M_i | P_i) \Gamma_i' S_{i+1} (M_{i+1}|P_{i+1})$.
As a consequence, from $\Gamma_i S_i (M_i | P_i)  \Gamma_i'$ one may continue into either $(M_i |Q_i \ P_i) R_i$ or $S_{i+1} (M_{i+1} | P_{i+1})$.
By \reffact{qf:5} we know that the sets $\proc((M_i | Q_i \setminus P_i) R_i)$ and $\proc(S_{i+1} (M_{i+1} | P_{i+1}))$ do not intersect.
Applying \reflem{tools:flp}, we obtain the run: $\Gamma_i S_i (M_i | P_i)  \Gamma_i' S_{i+1}(M_{i+1} | P_{i+1}) (M_i | Q_i \setminus P_i) R_i$.

Next, observe that from  $\Gamma_i S_i (M_i | P_i) \Gamma_i'S_{i+1} (M_{i+1}| P_{i+1})$ both $(M_i | Q_i \setminus P_i) R_i$ and $(M_{i+1} | Q_{i+1} \setminus P_{i+1}) R_{i+1}$ are possible.
From \reffact{qf:1} and \reffact{qf:5}, the sets $(Q_i \setminus P_i) \union \{ \coordi \}$ and  $(Q_{i+1} \setminus P_{i+1}) \union \{ \coordin \}$ are disjoint.
By \reflem{tools:flp}, $\sigma_{i+1} = \Gamma_{i+1} S_{i+1} (M_{i+1} | P_{i+1}) \Gamma'_{i+1} (M_{i+1} | Q_{i+1} \setminus P_{i+1}) R_{i+1}$ is a nice run , where
\begin{inparaenumorig}[]
\item $\Gamma_{i+1} = \Gamma_i S_i (M_i | P_i) \Gamma'_i$, and
\item $\Gamma'_{i+1} =  (M_i | Q_i \setminus P_i) R_i$.
\end{inparaenumorig}

\subsubsection{Correctness}
\labappendix{chaining:step:correctness}

The claims that follow establish the correctness of the above construction with respect to the properties at rank $i+1$.

\begin{claim}
  \labclaim{chaining:1}
  $\proc(\Gamma_{i+1}') \inter Q_{i+1} = P_{i+1}$
\end{claim}

\begin{proof}
  By construction, $\proc(\Gamma_{i+1}') = (Q_{i} \setminus P_i) \union \coordi$.
  On the other hand by \reffact{qf:2},  $P_{i+1} = Q_{i} \inter Q_{i+1}$.
  The result follows from \reffact{qf:5}.
\end{proof}

\begin{claim}
  \labclaim{chaining:2}
  $(S_{i+1} (M_{i+1}|P_{i+1}) (\Gamma_{i+1}' | \procSet \setminus \{q_{i+1}\})) \lmove (\Gamma_{i+1}'|q_{i+1})$
\end{claim}

\begin{proof}
  By construction, $\Gamma'_{i+1} = (M_i | Q_i \setminus P_i) R_i$.
  Applying \reffact{qf:4}, $q_{i+1} \notin Q_{i+1} \union \{\coordi\}$.
  Hence, $(\Gamma'_{i+1} | q_{i+1}) = (M_i | q_{i+1})$.
  Moreover, these steps left-move with
  \begin{inparaenumorig}[]
  \item $ (\Gamma_{i+1}' | \procSet \setminus \{q_{i+1}\})$ by \reflem{roll:1} and 
  \item with $(S_{i+1} (M_{i+1}|P_{i+1})$ since $q_{i+1} \notin Q_{i+1}$. 
  \end{inparaenumorig}
\end{proof}

\begin{claim}
  \labclaim{chaining:3}
  $\sigma_{i+1} \yield \cmdin \rightarrow \cmdi$
\end{claim}

\begin{proof}
  We now prove that a chaining effect occurs between the two commands $\cmdi$ and $\cmdin$ in $\sigma_{i+1}$.
  To this end, we first construct a nice run $\sigma$ in which $\cmd{i+1}$ is committed with $\cmd{i} \notin \deps(\cmd{i+1})$.
  Then, we establish that $\sigma$ is indistinguishable from $\sigma_{i+1}$ to the coordinator of $\cmdi$.
  This implies that $\coordi$ must add $\cmdin$ to the dependencies of $\cmdi$ in $\sigma_{i+1}$.

  With more details, our reasoning is as follows.
  First, consider $\Gamma_i S_i (M_i | P_i) \Gamma_i'$ that prefixes $\sigma_{i}$.
  Since $\coordin=q_i$ and $S_i (M_i|P_i) \lmove \Gamma_i'|q_i$ by our induction hypothesis, this run is equivalent to $\Gamma_i (\Gamma_i' | \coordin) S_i (M_i | P_i) (\Gamma_i' | \procSet \setminus \coordin)$.
  From which the prefix $\Gamma_i (\Gamma_i' | \coordin)$ is obtained.
  Applying \reflem{roll:1} to both $\cmd{i}$ and $\cmdin$ leads to $\sigma' \equaldef \Gamma_i (\Gamma_i' | \coordin) (SMR)_{i+1} (SMR)_i$.
  This run clearly satisfies that $\cmdin \rightarrow \cmdi$.

  The run $\sigma$ is then derived from a series of rewriting of $\sigma'$.
  Applying \reflem{roll:1} to $(SMR)_i$ then \reffact{qf:5} leads to $\Gamma_i (\Gamma_i' | \coordin) S_i (M_i | P_i) (SMR)_{i+1} (M_i | Q_i \setminus P_i) R_i$. 
  Following the same approach, $\Gamma_i (\Gamma_i' | \coordin) S_i (M_i | P_i) S_{i+1} (M | P_{i+1}) (M_i | Q_i \setminus P_i) R_i (M | Q_{i+1} \setminus P_{i+1}) R_{i+1}$
  The run $\sigma$ is then defined as $\Gamma_i (\Gamma_i' | \coordin) S_i (M_i | P_i) S_{i+1} (M | P_{i+1}) (M_i | Q_i \setminus P_i) R_i$.
  We observe that, since $\coordi$ takes the same steps in both $\sigma$ and $\sigma'$, then $\cmdin \rightarrow \cmdi$ holds in $\sigma$.

  Let us then examine the steps of $\coordi$ in the run $\sigma_{i+1}$.
  To this end, consider $\Gamma_i S_i (M_i | P_i) \Gamma_i'  S_{i+1} (M | P_{i+1}) (M_i | Q_i \setminus P_i) R_i$ that prefixes $\sigma_{i+1}$.
  By the assumption hypothesis $\mathfrak{P}(i)$, the steps $\Gamma_i'|\coordin$ left-move with $S_i (M_i | P_i) (\Gamma_i'|\procSet \setminus \coordin)$.
  This leads to the run $\Gamma_i (\Gamma_i' | \coordin) S_i (M_i | P_i) (\Gamma_i' | \procSet \setminus \{\coordin\})  S_{i+1} (M | P_{i+1}) (M_i | Q_i \setminus P_i) R_i$
  Applying again the assumption hypothesis, $\proc(\Gamma_i' | \procSet \setminus \{\coordin\}) \inter Q_i = \emptySet$.
  This leads to the following equivalent run $\Gamma_i (\Gamma_i' | \coordin) S_i (M_i | P_i) S_{i+1} (M | P_{i+1}) (M_i | Q_i \setminus P_i) R_i (\Gamma_i' | \procSet \setminus \{\coordin\})$.
  We pick the prefix $\Gamma_i (\Gamma_i' | \coordin) S_i (M_i | P_i) S_{i+1} (M | P_{i+1}) (M_i | Q_i \setminus P_i) R_i$.
  In this prefix, $\coordi$ takes the same steps as in $\sigma$.
  Hence, $\cmdin \rightarrow \cmdi$ holds in $\sigma_{i+1}$.
\end{proof}

\begin{claim}
  \labclaim{chaining:4}
  $\pending(\sigma_{i+1}) = \emptySet$
\end{claim}

\begin{proof}
  By induction, no message is pending in $\sigma_i$.
  Consider then a message sent in $(SMR)_{i+1}$.
  By Load-Balancing, this message is received in this sequence of steps.
\end{proof}

\begin{claim}
  \labclaim{chaining:5}
  $\forall \rho \in \cpaths(\Gamma_{i+1}') \ldotp \cardinalOf{\rho} \leq 1$
\end{claim}

\begin{proof}
  This claims follows from the definition of $\Gamma_{i+1}'$.
\end{proof}

\begin{claim}
  \labclaim{chaining:6}
  $\sigma_{i+1} \in \async{2}$
\end{claim}

\begin{proof}


  Assume that a message $m$ is sent in $\sigma_{i+1}$.
  From \refclaim{chaining:5}, $m$ is received in $\sigma_{i+1}$.
  Below, we conduct a case analysis depending on the position of $\recv(m)$ in $\sigma_{i+1}$.
  Our analysis shows that any path $\rho$ concurrent to $m$ is at most of length two.
  \begin{itemize}
    
  \item ($\Gamma_{i+1}$)
    If $\rho$ is concurrent to $m$ in $\sigma_{i+1}$, it is also concurrent to $m$ in $\sigma_i$.
    The induction hypothesis implies that $\cardinalOf{\rho} \leq 2$.

  \item ($S_{i+1} (M_{i+1} | P_{i+1}$)
    $\sigma_{i}$ does not contain a pending message.
    Thus, $m$ is sent by $\coord{i+1}$ in $S_{i+1}$.
    Applying \reflem{roll:1}, $\latency(M_{i+1}) = 0$ implies that $\cardinalOf{\rho}=1$.

  \item ($\Gamma_{i+1}'$)      
    By Load-Balancing, $m$ is sent in $(SMR)_i$.
    If $\rho$ ends in $\Gamma_i'$, then by the induction hypothesis $\rho$ is at most of size 2.
    Otherwise, the last message in $\rho$ is sent in $S_{i+1} (M_{i+1} | P_{i+1})$ by the coordinator of $\cmd{i+1}$.
    In that case, $(\Gamma_i'|\coord{i+1}) \lmove S_{i} (M_{i} | P_{i}) (\Gamma_i'|\procSet \setminus \coord{i+1})$ implies that $\cardinalOf{\rho}=1$.
    
  \item ($(M_{i+1} | Q_{i+1} \setminus P_{i+1}) R_{i+1}$)  
    $\send(m)$ occurs during the same sequence of steps, or it happens in $S_{i+1} (M_{i+1} | P_{i+1})$.
    The former case leads to $\cardinalOf{\rho}=1$.
    In the later, as no message sent in $\Gamma_{i+1}'$ is pending, $\rho$ is fully included in $S_{i+1} (M_{i+1} | P_{i+1}) \Gamma_{i+1}$.
    By \refclaim{chaining:5}, every such path in is (at most) of size 2.
    
  \end{itemize}

\end{proof}

\fi
\end{document}